\documentclass{article}
\usepackage{graphicx}
\usepackage{subfigure}
\usepackage{bm}
\usepackage{amsmath}
\usepackage{amsthm}
\usepackage{amsfonts}
\usepackage{fullpage}
\usepackage{nicefrac}
\usepackage{tikz}
\usetikzlibrary{decorations.markings}
\usepackage{float}
\usepackage{biblatex}
\newtheorem{definition}{Definition}
\newtheorem{theorem}{Theorem}
\newtheorem*{theorem*}{Theorem}
\newtheorem{corollary}{Corollary}
\newtheorem*{corollary*}{Corollary}

\newcommand{\bup}[2]{b_\uparrow(#1,#2)}
\newcommand{\bdown}[2]{b_\downarrow(#1,#2)}
\newcommand{\wup}[2]{w_\uparrow(#1,#2)}
\newcommand{\wdown}[2]{w_\downarrow(#1,#2)}

\newcommand{\Bup}[2]{B_\uparrow(#1,#2)}
\newcommand{\Bdown}[2]{B_\downarrow(#1,#2)}
\newcommand{\Wup}[2]{W_\uparrow(#1,#2)}
\newcommand{\Wdown}[2]{W_\downarrow(#1,#2)}

\newcommand{\Kint}{K_{\text{int}}}

\newcommand{\Bupw}[1]{B^{\uparrow}_\omega(#1)}
\newcommand{\Bdownw}[1]{B^{\downarrow}_\omega(#1)}
\newcommand{\Wupw}[1]{W^{\uparrow}_\omega(#1)}
\newcommand{\Wdownw}[1]{W^{\downarrow}_\omega(#1)}

\newcommand{\G}[2]{G_{#1}(#2,n_0;\omega)}

\newcommand{\up}{\uparrow}
\newcommand{\down}{\downarrow}

\newcommand{\z}[1]{z_{#1}}

\newcommand{\rp}[1]{r_{#1,+}}
\newcommand{\rmi}[1]{r_{#1,-}}
\newcommand{\rpm}[1]{r_{#1,\pm}}

\title{Dimer model on the square lattice with interface}
\author{Meredith Shea}
\date{}

\begin{document}

\maketitle

\begin{abstract}
    In this exposition, we consider the dimer problem on an infinite square lattice with partially non-periodic edge weights, which we refer to as the square lattice with interface. In particular, we compute an exact integral form of the inverse Kasteleyn operator and study its asymptotics behavior in different regions of the lattice to gain a general understanding of the local statistics of the model.
\end{abstract}

\tableofcontents

\section{Introduction}
A \textit{dimer configuration}, or perfect matching, of a graph is a subset of the edges which covers every vertex exactly once. Given a graph, one assigns a positive edge weight to each edge. The probability of a random dimer configuration is proportional to the product of all the edge weights of the edges included in the configuration. A \textit{dimer model} studies the random dimer configurations of a graph. In this paper, we consider the dimer model for a partially non-periodic weighting of the bipartite square lattice. In particular, we study the asymptotic behavior of the inverse Kasteleyn operator for such a system.

For case of uniform edge weights, in other words each dimer is equally likely, solutions to the dimer model on regions of the square and hexagonal lattice are well understood . Moreover, in the case of periodic weightings, the dimer model has been studied for many cases of the square lattice \cite{Kenyon97, ChhitaEtAl15, ChhitaJohansson16}. There are also similar results regarding isoradial graphs \cite{Kenyon02, Kenyon02, Li17}, and the more general Rail Yard graphs \cite{BoutillierEtAl17}. In general, less is known about dimer models when the weightings are not periodic. In some instances, non-periodic weightings can still be expressed as a \textit{Schur process} and exact solutions are known. Some examples of this can be seen in \cite{OkounkovReshetikhin03, Borodin07, BufetovGorin18, BoutillierEtAl17}.

Of principle interest in the study of dimer models is understanding the asymptotic behavior of the model. Specifically, when we allow the model to become large and re-scale appropriately, the \textit{height function} (a function associated to the faces of the graph for a given dimer configuration) tends to a deterministic limit shape which describes the global behavior of the model \cite{KenyonEtAl06, KenyonOkounkov07}. The limit shape of a model illustrates the different phases that occur. In particular, the phases possible are deterministic (frozen), critical (liquid or rough), and non-critical (gas or smooth). Depending on the weights in a given dimer model, the model can exhibit a single phase or a mix of some or all the phases. 

A classic example of the limit shape of a dimer model is the arctic curve phenomenon of the uniform Aztec diamond \cite{CohnEtAl96, Johansson05}. For this model, an elliptic curve separates a deterministic (the outside) and critical region (the interior). In the critical region, the height fluctuations are known to converge to a \textit{Gaussian free field} \cite{Kenyon01, Sheffield07}. More complicated limit shapes of the Aztec diamond have been studied in \cite{ChhitaEtAl15, ChhitaJohansson16}.

Another way of determining the phase in a certain region of the model is to consider the decay of the \textit{correlation function} between two dimers. In particular, the correlation functions are deterministic in the deterministic regions, decay inversely with respect to the square distance in the critical regions, and decay exponentially with respect to the distance in the non-critical regions. While the correlation functions are local in nature, they give an alternate method for understanding the phases and limit shape of the dimer model. Also associated to the dimer model is a weighted adjacency matrix known as the \textit{Kasteleyn operator} (or matrix). The correlation functions of a model are proportional to the appropriate entries of the inverse Kasteleyn operator. Thus, one way to partially understand the limit shape of the model is to understand the inverse Kasteleyn operator. This will be the direction we will take in this exposition. 

In this paper we will consider the dimer model on the infinite square lattice with, what I refer to as, an interface weighting. This weighting is non-periodic in the horizontal direction, and is formally defined in section 4.1. The main result of this paper is as integral form of the inverse Kasteleyn operator for this dimer model, which is presented in Theorem 1. We then go on to study the asymptotic behavior of this operator in order to gain an understanding of the limit shape. 

The organization of the paper is as follows. In section 2, we will discuss, in greater detail, necessary preliminary information and the Kasteleyn operator. Then, in section 3, we will go over some well known solutions in the periodic case that will be important for understanding the lattice weights chosen in section 4. We will then introduce the interface lattice and compute an exact integral form of the inverse Kasteleyn operator for the lattice in section 4. Lastly, in section 5 we use asymptotic methods to describe the local behavior of the Kasteleyn operator in various regions of the lattice. 
\subsection{Acknowledgements} I would like to thank Nicolai Reshetikhin, Sylvie Corteel, and Matthew Nicoletti for invaluable conversations about dimer models. This work was also made possible by NSF grants DMS-1902226 and DMS-2000093.

\section{Preliminaries}
In this section, we aim to define the Kasteleyn operator on a simple, bipartite graph (possibly infinite). While the graph need not be bipartite to define the Kasteleyn operator, we will only consider such cases in this exposition. 

Let $G=(B,W,E)$ be any planar bipartite graph, where $E$ denotes the set of edges and we partition the vertices into the set of black, $B$, and white, $W$, vertices. We first equip our graph with edge weights. For any edge, $e \in E$, we assign a weight $wt(e) \in \mathbb{R}_{>0}$. Suppose the edge $e$ connects the vertex $v$ to the vertex $u$, then we use the notation $wt(e) = wt(u,v) = wt(v,u)$ to denote the edge weight interchangeably. There is no direction associated to the edge weight.

Next, we assign an orientation to our graph, which we call a Kasteleyn orientation. For the purpose of this exposition, it suffices to exhibit a Kasteleyn orientation on the square lattice. Figure \ref{fig:kasteleynorientation} gives an example of a Kasteleyn orientation on a portion of the square lattice. Moreover, this will be our choice of orientation throughout this paper. 
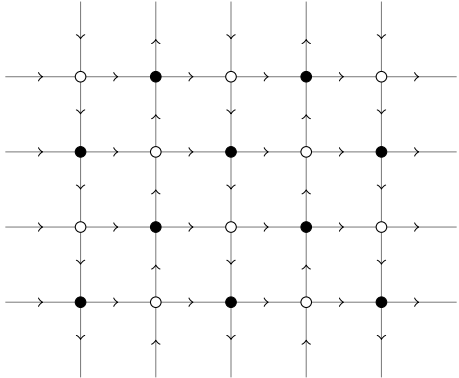
\begin{figure}[h] 
    \centering
    \begin{tikzpicture}
        \foreach \y in {0,...,3}
            {\draw[gray] (-3,\y) -- (3,\y);}
        \foreach \x in {-2,...,2}
            {\draw[gray] (\x,4) -- (\x,-1);}
        \foreach \x in {-2,0,2}
            \foreach \y in {1,3}
            {\filldraw[color=black,fill=white] (\x,\y) circle (2pt);}
        \foreach \x in {-1,1}
            \foreach \y in {2,0}
                {\filldraw[color=black,fill=white] (\x,\y) circle (2pt);}
        \foreach \x in {-1,1}
            \foreach \y in {1,3}
                {\filldraw[black] (\x,\y) circle (2pt);}
        \foreach \x in {-2,0,2}
            \foreach \y in {2,0}
                {\filldraw[black] (\x,\y) circle (2pt);}
        \foreach \x in {-3,...,2}
            \foreach \y in {0,...,3}
                {\path[tips, ->] (\x,\y) -- (\x+1/2,\y);}
        \foreach \x in {-2,0,2}
            \foreach \y in {4,...,0}
                {\path[tips,->] (\x,\y) -- (\x,\y-1/2);}
        \foreach \x in {-1,1}
            \foreach \y in {3,...,-1}
                {\path[tips,->] (\x,\y) -- (\x,\y+1/2);}
    \end{tikzpicture}
    \caption{\label{fig:kasteleynorientation} Example of a Kasteleyn orientation on the square lattice.}
\end{figure}
Given a Kasteleyn orientation of a graph $G$ we define the function $s: V \times V \to \mathbb{R}$ by,
\begin{equation}
    s(v,u) = 
    \begin{cases}
    0 & \text{if there is no edge between } v \text{ and } u \\
    1 & \text{if there is an edge between } v \text{ and } u \text{ orientated from } v \text{ to } u \\
    -1 & \text{if there is an edge between } v \text{ and } u \text{ orientated from } u \text{ to } v 
    \end{cases}
\end{equation}
Now we have all the information necessary to define the Kasteleyn operator.
\begin{definition}[Kasteleyn operator]
given a graph $G$, the Kasteleyn operator $K(G;wt,s):V \to V$ is defined by, 
\begin{equation}
    K(G;wt,s)v = \sum_{u\sim v} s(v,u)wt(v,u)u
\end{equation}
where the sum is taken over all neighboring vertices. 
\end{definition}
We will sometimes simplify the notation and use $K:=K(G;wt,s)$ when the choice of graph, edge weights, and orientation is clear. It is worth noting that two different choices of Kasteleyn orientation will result in Kasteleyn operators that are gauge equivalent. 

When $G$ is bipartite, it is useful to think of the Kasteleyn operator as acting on the space $B \oplus W$, where it has the structure
\begin{equation}
    K(G;wt,s) = \begin{pmatrix} 0 & -\tilde{K}(G;wt,s)^t \\ \tilde{K}(G;wt,s) & 0\end{pmatrix}
\end{equation}
where $\tilde{K}(G;wt,s): B \to W$. Again we will use the notation, $\tilde{K} := \tilde{K}(G;wt,s)$, when the choice of graph, edge weights, and orientation is clear. By choosing an identification of $B \simeq W$, for example a reference perfect matching, we can consider objects like the determinant of $\tilde{K}$. With this in mind, it suffices to consider just the operator $\tilde{K}$ when studying objects like the determinant and inverse of $K$.

The Kasteleyn operator is a central object in the study of dimer configurations (analogously tilings) on bipartite graphs. The following theorem is a central result in the literature, 
\begin{theorem*}[Kasteleyn \cite{Kasteleyn67}, Temperley and Fisher \cite{TemperleyFisher61}] \label{detkasteleyn}
For a bipartite graph with $wt(e) = 1$ for all $e \in E$ and any Kasteleyn orientation, the number of perfect matchings is equal to $\sqrt{\det K}$ (where we take the positive root). Equivalently, the number of perfect matchings is equal to $|\det \tilde{K}|$.
\end{theorem*}
For graphs with arbitrary edge weights, the Kasteleyn operator and its inverse also give us information on the probability that a set of edges are simultaneously covered in an arbitrary dimer model. This is detailed in the following theorem,
\begin{corollary*}[Kenyon \cite{Kenyon97}] \label{dimercorrelations}
Given a set of edges of a bipartite planar graph, $X = \{ e(w_1,b_1), \dots, e(w_k,b_k)\}$, the probability that all edges in $X$ are covered in a given dimer covering is,
\begin{equation}
    \left( \prod_{i=1}^k K(b_i,w_i) \right) \det \left( K^{-1}(w_i,b_j) \right)_{1 \leq i,j \leq k}
\end{equation}
where the coefficient $K(b,w)$ is defined by, $\tilde{K}b = \sum_w K(b,w)w$ (and likewise for the inverse coefficients).
\end{corollary*}
Given a set of two edges, we define their dimer correlation function to be the difference between their joint probability (of being covered) and the product of their individual probabilities (of being covered). If we consider the two edges to be $e(w_1,b_1)$ and $e(w_2,b_2)$ then, by the corollary above, the correlation function is proportional to $K^{-1}(w_1,b_2)K^{-1}(w_2,b_1)$. 

In general, as two edges become further away from each other in a model, we expect them to become less correlated. The components of the inverse Kasteleyn operator inform us on how quickly the correlation functions decay. All of these facts show us that one way to understand the local statistics of a dimer model is to understand the entries of the inverse Kasteleyn operator. This will be aim in sections 4 and 5 of this exposition. 
\section{Periodic weights and the Kasteleyn operator}
We will consider the Kasteleyn operator on the infinite square lattice with periodic weights and periodic boundary conditions. In these instance, it is useful to define the fundamental domain of the lattice. Utilizing the fundamental domain will allow us to understand local dimer statistics of the model without having to go through what are well-known, but often long, computations. 
\subsection{Uniform weights}
First we will consider the uniform weights, $wt(e) = 1$ for all $e \in E$. Figure \ref{fig:fundamentaldomainuniform} depicts the $2 \times 2$ fundamental domain of this lattice. The fundamental domain consists of two black vertices vertices, $\bup{n}{m}$ and $\bdown{n}{m}$, and two white vertices, $\wup{n}{m}$ and $\wdown{n}{m}$. The coordinates $(n,m)$ denote the particular fundamental domain the vertices belong to.
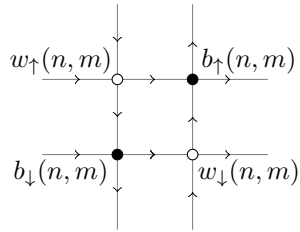
\begin{figure}[h] 
    \centering
    \begin{tikzpicture}
        \draw[gray] (0,2) -- (0,-1);
        \draw[gray] (1,2) -- (1,-1);
        \draw[gray] (-1,1) -- (2,1);
        \draw[gray] (-1,0) -- (2,0);
        \filldraw[black] (1,1) circle (2pt);
        \node at (1.75,1.25) {$b_\uparrow(n,m)$};
        \filldraw[black] (0,0) circle (2pt);
        \node at (-0.75,-0.25) {$b_\downarrow(n,m)$};
        \filldraw[color=black,fill=white] (1,0) circle (2pt);
        \node at (-0.75,1.25) {$w_\uparrow(n,m)$};
        \filldraw[color=black,fill=white] (0,1) circle (2pt);
        \node at (1.75,-0.25) {$w_\downarrow(n,m)$};
        \path[tips, ->] (0,1) -- (1/2,1);
        \path[tips, ->] (-1,1) -- (-1/2,1);
        \path[tips, ->] (1,1) -- (1+1/2,1);
        \path[tips, ->] (0,0) -- (1/2,0);
        \path[tips, ->] (0,0) -- (1/2,0);
        \path[tips, ->] (-1,0) -- (-1/2,0);
        \path[tips, ->] (1,0) -- (1+1/2,0);
        \path[tips, ->] (0,2) -- (0,1+1/2);
        \path[tips, ->] (0,1) -- (0,1/2);
        \path[tips, ->] (0,0) -- (0,-1/2);
        \path[tips, ->] (1,-1) -- (1,-1/2);
        \path[tips, ->] (1,0) -- (1,1/2);
        \path[tips, ->] (1,1) -- (1,1+1/2);
    \end{tikzpicture}
    \caption{\label{fig:fundamentaldomainuniform} The fundamental domain of the lattice with uniform edge weights.}
\end{figure}

Our goal is to compute the inverse Kasteleyn operator of this fundamental example. It will suffice to just compute the inverse of $\tilde{K}$. We start off by computing  how $\tilde{K}$ acts on the two black vertices of an arbitrary fundamental domain,
\begin{align}
    \tilde{K}\bup{n}{m} &= \wup{n+1}{m} - \wup{n}{m} + \wdown{n}{m+1} - \wdown{n}{m} \label{eq:uniformkbup}\\
    \tilde{K}\bdown{n}{m} &= -\wup{n}{m} + \wup{n}{m-1} + \wdown{n}{m} - \wdown{n-1}{m} \label{eq:uniformkbdown}
\end{align}
Since the lattice is translationally invariant in the $n$ and $m$ directions, we can Fourier transform equations \eqref{eq:uniformkbup} and \eqref{eq:uniformkbdown}. We will let $z$ and $\omega$ denote the Fourier variables in the $n$ and $m$ directions, respectively. We define the Fourier transform of the $b_i(n,m)$ and $w_i(n,m)$ functions as, 
\begin{align}
    B_i(z,\omega) &= \sum_{n,m} b_i(n,m)z^{-n}\omega^{-m} \\
    W_i(z,\omega) &= \sum_{n,m} w_i(n,m)z^{-n}\omega^{-m}
\end{align}
where $i = \uparrow, \; \downarrow$. Plugging these into equations \eqref{eq:uniformkbup} and \eqref{eq:uniformkbdown} gives, 
\begin{align}
    \tilde{K}_{z,\omega}\Bup{z}{\omega} &= (z-1)\Wup{z}{\omega} + (\omega-1)\Wdown{z}{\omega}\\
    \tilde{K}_{z,\omega}\Bdown{z}{\omega} &= (\omega^{-1}-1)\Wup{z}{\omega} + (1-z^{-1})\Wdown{\omega}{z}
\end{align}
Using the above, we can write a matrix form of $\tilde{K}_{z,w}$,
\begin{equation}
    \tilde{K}_{z,w} = \begin{pmatrix} z-1 & w^{-1}-1 \\ w-1 & 1-z^{-1}\end{pmatrix}
\end{equation}
Inverting the above matrix directly gives, 
\begin{equation}
    \left(\tilde{K}_{z,\omega}\right)^{-1} = \frac{1}{z+z^{-1}+w+w^{-1}-4}\begin{pmatrix}
    1-z^{-1} & 1 - w^{-1} \\ 1-w & z-1 
    \end{pmatrix} \label{eq:kunifinvFT}
\end{equation}
where $\left(\tilde{K}_{z,\omega}\right)^{-1}: W \to B$. Using the inverse Fourier transform will yield a double integral form for the coefficients of the inverse Kasteleyn. Since the lattice is translationally invariant, it suffices to fix one of the vertices at the fundamental domain $(0,0)$. Let's take a look at the integral form of $\tilde{K}^{-1}\left(\wup{0}{0},\bdown{n}{m}\right)$,
\begin{equation} \label{eq:kintunifasymp}
    \tilde{K}^{-1}\left(\wup{0}{0},\bdown{n}{m}\right) = \frac{1}{4 \pi^2} \int_{|z|=1} \int_{|w| = 1} \frac{(1-w)z^{n}w^{m}}{z+z^{-1}+w+w^{-1}-4}\frac{dw}{iw}\frac{dz}{iz}
\end{equation}
Choosing different vertices of the fundamental domain would simply change the numerator of the integrand according to the appropriate entry of equation \eqref{eq:kunifinvFT}.

Next, we consider the asymptotic behavior of $\tilde{K}^{-1}\left(w_\up(0,0),b_\down(n,m)\right)$ as $n$ and $m$ become large. We will first approach this using integral methods. Let $\omega = e^{i\theta}$ and rewrite the denominator in \eqref{eq:kintunifasymp} as, 
\begin{equation}
    (z+z^{-1}+w+w^{-1} - 4)z = (z - \lambda_-)(z+\lambda_+)
\end{equation}
where,
\begin{equation*}
    \lambda_{\pm} = \frac{1}{2}\left(4-2\cos \theta \pm \sqrt{(2\cos\theta-4)^2-4}\right)
\end{equation*}
We then use the residue theorem to handle the integral with respect to $z$, and we are left with an integral with respect to $\theta$. We use asymptotic techniques on this integral and we find that $\tilde{K}^{-1}\left(\wup{0}{0},\bdown{n}{m}\right)$ decays inversely with respect to the distance. We refer to this lattice as \textit{critical} because of this decay rate.  
\subsection{The spectral curve of the dimer model}
We can see from the above analysis of $\tilde{K}^{-1}\left(\wup{0}{0},\bdown{n}{m}\right)$ that the asymptotic behavior depends on whether the roots, $\lambda_{\pm}$, lie inside or outside the unit circle. Moreover, these roots appear in the denominator of equation \eqref{eq:kintunifasymp} which is, in fact, the determinant of $\tilde{K}_{z,w}$ after a choice of identification $B \simeq W$. We choose the identification which associates $b_{\up}(n,m) \simeq w_{\up}(n,m)$ and $b_{\down}(n,m) \simeq w_{\down}(n,m)$ for all $n$ and $m$. With this in mind we compute, 
\begin{equation}
    \det \tilde{K}_{z,w} = p(z,w) = z+z^{-1}+w+w^{-1}-4
\end{equation}
We call the function $p(z,w)$ the \textit{spectral curve} of the dimer model because it's roots tell us about the spectrum of the Kasteleyn operator and consequently the behavior of the correlation functions \cite{KenyonEtAl06}. We refer to the lattice model as being \textit{critical} (rough) if the entries of the inverse Kasteleyn operator decay inversely with respect to distance. We say that the lattice is \textit{non-critical} (smooth) if the entries of the inverse Kasteleyn operator decay exponentially with respect to distance. 

For periodic weightings and boundary conditions, we can determine whether the lattice model is critical by simply observing the roots of the spectral curve. If the roots of $p(z,w)$ lie on the unit torus, then the lattice is critical, otherwise the lattice is non-critical. In the next section we will look at a richer example to illustrate this fact more completely.  
\subsection{A periodic weighting}
Let's now consider the periodic weighting given by the fundamental domain in figure \ref{fig:fundamentaldomainperiodic}. Using the process described in section 3.1, we can write the $2\times 2$ matrix $\tilde{K}_{z,\omega}$,
\begin{equation}
    \tilde{K}_{z,\omega} = \begin{pmatrix}
    z-1 & b\omega^{-1}-a \\ a\omega-b & 1-z^{-1}
    \end{pmatrix}
\end{equation}
and compute the determinant using the same identification that was discussed in section 3.2,
\begin{equation}
    \det \tilde{K}_{z,\omega} = p(z,\omega) = -2 - 2ab + b^2 \omega^{-1} +a^2 \omega + z^{-1} + z \label{eq:pzw_ab}
\end{equation}
The asymptotic behavior of the entries of the inverse Kasteleyn depend on the roots of the above function. Moreover, the roots depend on the value of the positive parameters $a$ and $b$.
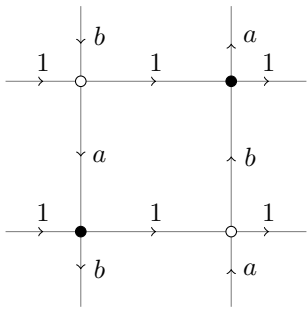
\begin{figure}[h] 
    \centering
    \begin{tikzpicture}
        \draw[gray] (0,3) -- (0,-1);
        \draw[gray] (2,3) -- (2,-1);
        \draw[gray] (-1,2) -- (3,2);
        \draw[gray] (-1,0) -- (3,0);
        \filldraw[black] (2,2) circle (2pt);
        \filldraw[black] (0,0) circle (2pt);
        \filldraw[color=black,fill=white] (2,0) circle (2pt);
        \filldraw[color=black,fill=white] (0,2) circle (2pt);
        \path[tips, ->] (-1,2) -- (-0.5,2);
        \path[tips, ->] (0,2) -- (1,2);
        \path[tips, ->] (2,2) -- (2.5,2);
        \path[tips, ->] (-1,0) -- (-0.5,0);
        \path[tips, ->] (0,0) -- (1,0);
        \path[tips, ->] (2,0) -- (2.5,0);
        \path[tips, ->] (0,3) -- (0,2.5);
        \path[tips, ->] (0,2) -- (0,1);
        \path[tips, ->] (0,0) -- (0,-0.5);
        \path[tips, ->] (2,-1) -- (2,-0.5);
        \path[tips, ->] (2,0) -- (2,1);
        \path[tips, ->] (2,2) -- (2,2.5);
        \node at (0.25,-0.5) {$b$};
        \node at (0.25,1) {$a$};
        \node at (0.25,2.6) {$b$};
        \node at (2.25,-0.5) {$a$};
        \node at (2.25,1) {$b$};
        \node at (2.25,2.6) {$a$};
        \node at (-0.5,0.25) {$1$};
        \node at (1,0.25) {$1$};
        \node at (2.5,0.25) {$1$};
        \node at (-0.5,2.25) {$1$};
        \node at (1,2.25) {$1$};
        \node at (2.5,2.25) {$1$};
    \end{tikzpicture}
    \caption{\label{fig:fundamentaldomainperiodic} A periodic weighting on the $2 \times 2$ fundamental domain. We label the domain vertices in the same manner as figure \ref{fig:fundamentaldomainuniform}.}
\end{figure}
If we fix the value of $z$ on the unit torus, $|p(z,\omega)|$ is minimal for $\omega = 1$. Applying this we are left with,  
\begin{equation}
    0 = (a-b)^2-2 + z+z^{-1}
\end{equation}
Thus, $p(z,\omega)$ has solutions on the unit torus if and only if $|a-b|<2$. So the lattice weights are critical when $|a-b|<2$ and non-critical when $|a-b| > 2$. The fact that the above lattice can be made critical or non-critical by changing the values of $a$ and $b$ is the motivation for the lattice we will describe in section 4.1. 
\section{Kasteleyn operator with interface}
In this section, we will introduce a new weighting on the infinite square lattice that is periodic only in one direction. We will refer to this lattice as the \textit{square lattice with interface} and to the corresponding Kasteleyn operator as the \textit{Kasteleyn operator with interface}. We will compute an integral form for the inverse Kasteleyn operator with interface. Later on, in section 5, we will consider the asymptotic behavior of this integral form to determine the local statistics of the lattice. 
\subsection{Square lattice with interface}
To understand the square lattice with interface, it will be useful to understand a coordinate system on the lattice. In the horizontal and vertical directions we will enumerate every other face of the lattice, such that each vertex is associated to exactly one enumerated face. The vertical direction will be known as the $m$-direction, while the horizontal direction will be known as the $n$-direction. Thus every enumerated face (and the four vertices associated to it) will have the coordinates $(n,m)$. These coordinates are depicted in figure \ref{fig:latticecoords}.
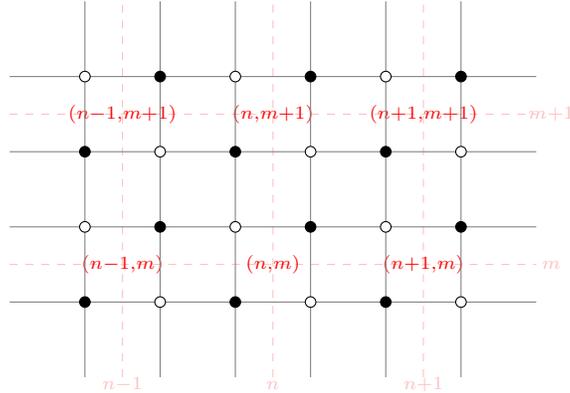
\begin{figure}[h] 
    \centering
    \begin{tikzpicture}
        \draw[pink,dashed] (-0.5,-1) -- (-0.5,4);
        \node[pink] at (-0.5,-1.1) {$\scriptstyle n$};
        \draw[pink,dashed] (-2.5,-1) -- (-2.5,4);
        \node[pink] at (-2.5,-1.1) {$\scriptstyle n-1$};
        \draw[pink,dashed] (1.5,-1) -- (1.5,4);
        \node[pink] at (1.5,-1.1) {$\scriptstyle n+1$};
        \draw[pink,dashed] (-4,0.5) -- (3,0.5);
        \node[pink] at (3.2,0.5) {$\scriptstyle m$};
        \draw[pink,dashed] (-4,2.5) -- (3,2.5);
        \node[pink] at (3.2,2.5) {$\scriptstyle m+1$};
        \node[red] at (-0.5,0.5) {$\scriptstyle (n,m)$};
        \node[red] at (-0.5,2.5) {$\scriptstyle (n,m+1)$};
        \node[red] at (-2.5,0.5) {$\scriptstyle (n-1,m)$};
        \node[red] at (1.5,0.5) {$\scriptstyle (n+1,m)$};
        \node[red] at (-2.5,2.5) {$\scriptstyle (n-1,m+1)$};
        \node[red] at (1.5,2.5) {$\scriptstyle (n+1,m+1)$};
        \foreach \x in {-3,...,2}
        {\draw[gray] (\x,-1) -- (\x,4);}
        \foreach \y in {0,...,3}
        {\draw[gray] (-4,\y) -- (3,\y);}
        \foreach \x in {-3,-1,1}
        {\filldraw[color=black,fill=white] (\x,3) circle (2pt);
        \filldraw[black] (\x+1,3) circle (2pt);}
        \foreach \x in {-3,-1,1}
        {\filldraw[black] (\x,2) circle (2pt);
        \filldraw[color=black,fill=white] (\x+1,2) circle (2pt);}
        \foreach \x in {-3,-1,1}
        {\filldraw[color=black,fill=white] (\x,1) circle (2pt);
        \filldraw[black] (\x+1,1) circle (2pt);}
        \foreach \x in {-3,-1,1}
        {\filldraw[black] (\x,0) circle (2pt);
        \filldraw[color=black,fill=white] (\x+1,0) circle (2pt);}
    \end{tikzpicture}
    \caption{\label{fig:latticecoords} Coordinates on the bipartite square lattice. Note that the face and the four adjacent vertices are all denoted by the given coordinate.}
\end{figure}
On the square lattice with interface, each enumerated face with $n \leq 0$ will have edge weights defined by figure \ref{fig:fundamentaldomainperiodic} and for enumerated faces with $n >0$ the edge weights will be defined by figure \ref{fig:funddomainn>0}. Since the change of weights occurs at $n = 0$, we call this line in the lattice \textit{the interface}. The depiction of the lattice near the interface is shown in figure \ref{fig:interfacelattice}.
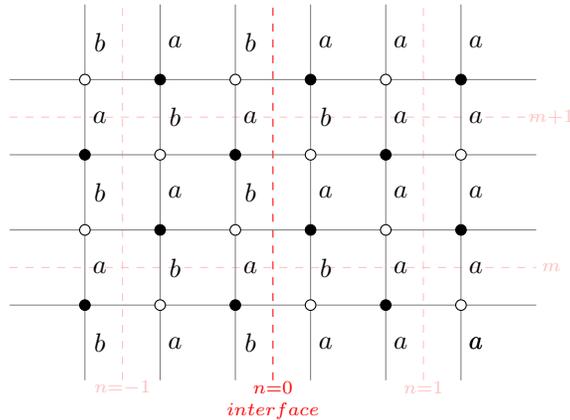
\begin{figure}[h] 
    \centering
    \begin{tikzpicture}
        \draw[red,dashed] (-0.5,-1) -- (-0.5,4);
        \node[red] at (-0.5,-1.1) {$\scriptstyle n=0$};
        \node[red] at (-0.5,-1.4) {$\scriptstyle interface$};
        \draw[pink,dashed] (-2.5,-1) -- (-2.5,4);
        \node[pink] at (-2.5,-1.1) {$\scriptstyle n=-1$};
        \draw[pink,dashed] (1.5,-1) -- (1.5,4);
        \node[pink] at (1.5,-1.1) {$\scriptstyle n=1$};
        \draw[pink,dashed] (-4,0.5) -- (3,0.5);
        \node[pink] at (3.2,0.5) {$\scriptstyle m$};
        \draw[pink,dashed] (-4,2.5) -- (3,2.5);
        \node[pink] at (3.2,2.5) {$\scriptstyle m+1$};
        \foreach \x in {-3,...,2}
        {\draw[gray] (\x,-1) -- (\x,4);}
        \foreach \y in {0,...,3}
        {\draw[gray] (-4,\y) -- (3,\y);}
        \foreach \x in {-3,-1,1}
        {\filldraw[color=black,fill=white] (\x,3) circle (2pt);
        \filldraw[black] (\x+1,3) circle (2pt);}
        \foreach \x in {-3,-1,1}
        {\filldraw[black] (\x,2) circle (2pt);
        \filldraw[color=black,fill=white] (\x+1,2) circle (2pt);}
        \foreach \x in {-3,-1,1}
        {\filldraw[color=black,fill=white] (\x,1) circle (2pt);
        \filldraw[black] (\x+1,1) circle (2pt);}
        \foreach \x in {-3,-1,1}
        {\filldraw[black] (\x,0) circle (2pt);
        \filldraw[color=black,fill=white] (\x+1,0) circle (2pt);}
        \node at (-2.8,3.5) {$\scriptsize b$};
        \node at (-2.8,2.5) {$\scriptsize a$};
        \node at (-2.8,1.5) {$\scriptsize b$};
        \node at (-2.8,0.5) {$\scriptsize a$};
        \node at (-2.8,-0.5) {$\scriptsize b$};
        \node at (-1.8,3.5) {$\scriptsize a$};
        \node at (-1.8,2.5) {$\scriptsize b$};
        \node at (-1.8,1.5) {$\scriptsize a$};
        \node at (-1.8,0.5) {$\scriptsize b$};
        \node at (-1.8,-0.5) {$\scriptsize a$};
        \node at (-0.8,3.5) {$\scriptsize b$};
        \node at (-0.8,2.5) {$\scriptsize a$};
        \node at (-0.8,1.5) {$\scriptsize b$};
        \node at (-0.8,0.5) {$\scriptsize a$};
        \node at (-0.8,-0.5) {$\scriptsize b$};
        \node at (0.2,3.5) {$\scriptsize a$};
        \node at (0.2,2.5) {$\scriptsize b$};
        \node at (0.2,1.5) {$\scriptsize a$};
        \node at (0.2,0.5) {$\scriptsize b$};
        \node at (0.2,-0.5) {$\scriptsize a$};
        \node at (1.2,3.5) {$\scriptsize a$};
        \node at (1.2,2.5) {$\scriptsize a$};
        \node at (1.2,1.5) {$\scriptsize a$};
        \node at (1.2,0.5) {$\scriptsize a$};
        \node at (1.2,-0.5) {$\scriptsize a$};
        \node at (2.2,-0.5) {$\scriptsize a$};
        \node at (2.2,3.5) {$\scriptsize a$};
        \node at (2.2,2.5) {$\scriptsize a$};
        \node at (2.2,1.5) {$\scriptsize a$};
        \node at (2.2,0.5) {$\scriptsize a$};
        \node at (2.2,-0.5) {$\scriptsize a$};
    \end{tikzpicture}
    \caption{\label{fig:interfacelattice} The planar, bipartite square lattice with interface at $n = 0$. All unlabeled edges have weight one. We also equipped this lattice with the Kasteleyn orientation depicted in figure \ref{fig:kasteleynorientation}.}
\end{figure}
\begin{figure}[h] 
    \centering
    \begin{tikzpicture}
        \draw[gray] (0,3) -- (0,-1);
        \draw[gray] (2,3) -- (2,-1);
        \draw[gray] (-1,2) -- (3,2);
        \draw[gray] (-1,0) -- (3,0);
        \filldraw[black] (2,2) circle (2pt);
        \filldraw[black] (0,0) circle (2pt);
        \filldraw[color=black,fill=white] (2,0) circle (2pt);
        \filldraw[color=black,fill=white] (0,2) circle (2pt);
        \path[tips, ->] (-1,2) -- (-0.5,2);
        \path[tips, ->] (0,2) -- (1,2);
        \path[tips, ->] (2,2) -- (2.5,2);
        \path[tips, ->] (-1,0) -- (-0.5,0);
        \path[tips, ->] (0,0) -- (1,0);
        \path[tips, ->] (2,0) -- (2.5,0);
        \path[tips, ->] (0,3) -- (0,2.5);
        \path[tips, ->] (0,2) -- (0,1);
        \path[tips, ->] (0,0) -- (0,-0.5);
        \path[tips, ->] (2,-1) -- (2,-0.5);
        \path[tips, ->] (2,0) -- (2,1);
        \path[tips, ->] (2,2) -- (2,2.5);
        \node at (0.25,-0.5) {$a$};
        \node at (0.25,1) {$a$};
        \node at (0.25,2.6) {$a$};
        \node at (2.25,-0.5) {$a$};
        \node at (2.25,1) {$a$};
        \node at (2.25,2.6) {$a$};
        \node at (-0.5,0.25) {$1$};
        \node at (1,0.25) {$1$};
        \node at (2.5,0.25) {$1$};
        \node at (-0.5,2.25) {$1$};
        \node at (1,2.25) {$1$};
        \node at (2.5,2.25) {$1$};
    \end{tikzpicture}
    \caption{\label{fig:funddomainn>0} The edge weights surrounding ennumerated faces when $n >0$.}
\end{figure}
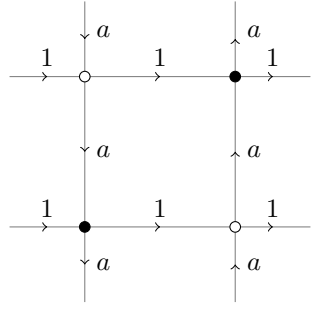
The aim of this 'interface' weighting is to choose an appropriate $a$ and $b$ such that the two halves of the lattice exhibit different fundamental behavior. In particular, we are interested in the case where $b-a > 2$, because one might expect half the weights to behave critically while the other half behaves non-critically. From this set up, we can explicitly define what we mean by the Kasteleyn operator with interface,
\begin{definition}[Kasteleyn operator with interface]
    Let $G$ be the square lattice with interface defined in figure \ref{fig:interfacelattice} equipped with the Kasteleyn orientation depicted in figure \ref{fig:kasteleynorientation}. We define the Kasteleyn operator with interface to be, 
    \begin{equation}
        \Kint := \tilde{K}(G)
    \end{equation}
\end{definition}
Again, since we are using the operator $\tilde{K}$, it is important to mention that we use the same identification of $B \simeq W$ as was described in section 3.2. Since we aim to compute the inverse of the Kasteleyn operator with interface, we need to comment on the boundary conditions of the lattice. In the $m$-direction we will assume periodic boundary conditions since the lattice in translationally invariant in this direction. In the $n$-directions, instead of applying typical boundary conditions we will apply conditions of convergence directly to the inverse Kasteleyn. These conditions will be seen in action in the proof of Theorem 1. 
\subsection{Inverse Kasteleyn with interface}
In this section we will state an explicit integral form of the entries of the inverse Kasteleyn operator with interface defined above. Before stating we first define the following functions with variable $\omega$,
\begin{align}
    z_1(\omega) &= a\omega-b \label{eq:z1def} \\
    z_2(\omega) &= a(\omega-1) \label{eq:z2def}
\end{align}
\begin{equation} \label{eq:rdef}
    r_{i,\pm}(\omega) = \frac{1}{2\omega}\left(2 \omega + z_i(\omega)^2 \pm z_i(\omega)\sqrt{z_i(\omega)^2-4\omega}\right)
\end{equation}
and the following vector valued function,
\begin{equation} \label{eq:vdef}
    \mathbf{v}_{i,\pm}(\omega) = \begin{pmatrix} \omega^{-1}z_i(\omega) \\ r_{i,\pm}(\omega)-1
    \end{pmatrix}
\end{equation}
where $r_{i,\pm}(\omega)$ and $\mathbf{v}_{i,\pm}(\omega)$ are defined for $i=1,2$. Although the above functions explicitly depend on $\omega$, we will simplify the notation where it does not create confusion. In these instances we will let, 
\begin{align*}
    z_i &:= z_i(\omega) \\
    r_{i,\pm} &:= r_{i,\pm}(\omega)\\
    \mathbf{v}_{i,\pm} &:= \mathbf{v}_{i,\pm}(\omega)
\end{align*}
With these functions in mind we can now define the functions central to the inverse Kasteleyn operator with interface, 
\begin{equation} \label{eq:g11g21sol1}
    \begin{pmatrix}
    G^>_{\up\up}(n,n_o;\omega) \\ G^>_{\down\up}(n,n_0;\omega)
    \end{pmatrix} = 
    \begin{cases}
    c_1(n_0;\omega) \mathbf{v}_{1,+}(\omega) r_{1,+}(\omega)^n & n \leq 0 \\
    c_2(n_0;\omega) \mathbf{v}_{2,+}(\omega) r_{2,+}(\omega)^n + c_3(n_0;\omega)\mathbf{v}_{2,-}(\omega) r_{2,-}(\omega)^n & 0 < n \leq n_0 \\
    c_4(n_0;\omega) \mathbf{v}_{2,-}(\omega) r_{2,-}(\omega)^n & n > n_0
    \end{cases}
\end{equation}
\begin{equation} \label{eq:g12g22sol1}
    \begin{pmatrix}
    G^>_{\up\down}(n,n_0;\omega) \\ G^>_{\down\down}(n,n_0;\omega)
    \end{pmatrix} = 
    \begin{cases}
    d_1(n_0;\omega) \mathbf{v}_{1,+}(\omega) r_{1,+}(\omega)^n & n \leq 0 \\
    d_2(n_0;\omega) \mathbf{v}_{2,+}(\omega) r_{2,+}(\omega)^n + d_3(n_0;\omega)\mathbf{v}_{2,-}(\omega) r_{2,-}(\omega)^n & 0 < n < n_0 \\
    d_4(n_0;\omega) \mathbf{v}_{2,-}(\omega) r_{2,-}(\omega)(\omega)^n & n \geq n_0
    \end{cases}
\end{equation}
\begin{equation} \label{eq:g11g21sol2}
    \begin{pmatrix}
    G^<_{\up\up}(n,n_0;\omega) \\ G^<_{\down\up}(n,n_0;\omega)
    \end{pmatrix} = 
    \begin{cases}
    c'_1(n_0;\omega) \mathbf{v}_{1,+}(\omega) r_{1,+}(\omega)^n & n \leq n_0 \\
    c'_2(n_0;\omega) \mathbf{v}_{1,+}(\omega) r_{1,+}(\omega)^n + c'_3(n_0;\omega)\mathbf{v}_{1,-}(\omega) r_{1,-}(\omega)^n & n_0 < n \leq 0 \\
    c'_4(n_0;\omega) \mathbf{v}_{2,-}(\omega) r_{2,-}(\omega)^n & n > 0
    \end{cases}
\end{equation}
\begin{equation} \label{eq:g12g22sol2}
    \begin{pmatrix}
    G^<_{\up\down}(n,n_0;\omega) \\ G^<_{\down\down}(n,n_0;\omega)
    \end{pmatrix} = 
    \begin{cases}
    d'_1(n_0;\omega) \mathbf{v}_{1,+}(\omega) r_{1,+}(\omega)^n & n < n_0 \\
    d'_2(n_0;\omega) \mathbf{v}_{1,+}(\omega) r_{1,+}(\omega)^n + d'_3(n_0;\omega)\mathbf{v}_{1,-}(\omega) r_{1,-}(\omega)^n & n_0 \leq n \leq 0 \\
    d'_4(n_0;\omega) \mathbf{v}_{2,-}(\omega) r_{2,-}(\omega)^n & n > 0
    \end{cases}
\end{equation}
where the coefficients $c_i(n_0;\omega)$, $c'_i(n_0;\omega)$, $d_i(n_0;\omega)$, and $d'_i(n_0;\omega)$ are defined in the Appendix A. The vertices $w_\up(n,m)$, $w_\down(n,m)$, $b_\up(n,m)$, and $b_\down(n,m)$ are the same vertices described in figure \ref{fig:fundamentaldomainuniform}.
\begin{theorem}[Inverse Kasteleyn operator with interface]
The inverse Kasteleyn operator has the form, 
\begin{equation}
    \Kint^{-1}\left(w_i(n_0,0),b_j(n,m)\right) = 
    \begin{cases} 
    \displaystyle\frac{1}{2\pi} \int_{|\omega|=1} G^>_{ij}(n,n_0;\omega)\omega^m\frac{d\omega}{i\omega} & n_0 > 0 \\[1.5em]
    \displaystyle\frac{1}{2\pi} \int_{|\omega|=1} G^<_{ij}(n,n_0;\omega)\omega^m\frac{d\omega}{i\omega} & n_0 < 0
    \end{cases}
\end{equation}
where $i = \up,\down$, $j=\up,\down$. Since the lattice is translationally invariant in the $m$-direction we fix $m_0 = 0$ without loss of generality.
\end{theorem}
\begin{proof}
First let us compute how $\Kint$ acts on the vertices $\bup{n}{m}$ and $\bdown{n}{m}$,
\begin{equation} \label{eq:kint1}
    \Kint \bup{n}{m} =
    \begin{cases}
    \wup{n+1}{m}-\wup{n}{m}+a\wdown{n}{m+1}-b\wdown{n}{m} & n \leq 0 \\
    \wup{n+1}{m}-\wup{n}{m}+a\wdown{n}{m+1}-a\wdown{n}{m} & n > 0
    \end{cases}
\end{equation}
\begin{equation} \label{eq:kint2}
    \Kint \bdown{n}{m} =
    \begin{cases}
    b\wup{n}{m-1}-a\wup{n}{m}+\wdown{n}{m}-\wdown{n-1}{m} & n \leq 0 \\
    a\wup{n}{m-1}-a\wup{n}{m}+\wdown{n}{m}-\wdown{n-1}{m} & n > 0
    \end{cases}
\end{equation}
Since we are applying periodic boundary conditions to the lattice in the $m$-direction, we can use the Fourier transform on the above. We will let $\omega$ to denote the Fourier variable and define,
\begin{align*}
    B_{\omega}^j(n) &= \sum_m b_j(n,m)\omega^{-m} \\
    W_{\omega}^j(n) &= \sum_m w_j(n,m)\omega^{-m} 
\end{align*}
for $j = \up,\down$. We can now rewrite equations \eqref{eq:kint1} and \eqref{eq:kint2} using the Fourier variable, 
\begin{equation}
    \Kint(n;\omega) \Bupw{n} = 
    \begin{cases}
    \Wupw{n+1} - \Wupw{n} + z_1(\omega)\Wdownw{n} & n \leq 0 \\
    \Wupw{n+1} - \Wupw{n} + z_2(\omega)\Wdownw{n} & n > 0
    \end{cases}
\end{equation}
\begin{equation}
    \Kint(n;\omega) \Bdownw{n} = 
    \begin{cases}
    -\omega z_1(\omega)\Wupw{n} + \Wdownw{n} - \Wdownw{n-1} & n \leq 0 \\
    -\omega z_2(\omega)\Wupw{n} + \Wdownw{n} - \Wdownw{n-1} & n > 0
    \end{cases}
\end{equation}
The functions $z_1(\omega)$ and $z_2(\omega)$ are defined in equations \eqref{eq:z1def} and \eqref{eq:z2def}. From the above equations we can write $\Kint(n;\omega)$ as a (piece-wise) difference operator acting on the vector $\begin{pmatrix}\Bupw{n} \\ \Bdownw{n} \end{pmatrix}$,
\begin{equation}
    \Kint(n;\omega) =
    \begin{cases}
    \begin{pmatrix} D_+ - 1 & -\omega z_1(\omega) \\ z_1(\omega) & 1-D_- \end{pmatrix} & n \leq 0 \\[1.5em]
    \begin{pmatrix} D_+ - 1 & -\omega z_2(\omega) \\ z_2(\omega) & 1-D_- \end{pmatrix} & n > 0 \\
    \end{cases}
\end{equation}
Where $D_\pm$ are the discrete operators defined by, $D_\pm f(n) = f(n\pm 1)$. We are now interested in finding the Green's function of the difference operator above. The Green's function will satisfy, 
\begin{equation} \label{eq:kgsystem}
    \Kint(n;\omega)\mathcal{G}(n,n_0;\omega) = \begin{pmatrix} \delta_{n,n_0} & 0 \\ 0 & \delta_{n,n_0} \end{pmatrix}
\end{equation}
where $\delta_{n,n_0}$ is the Kronecker delta. We will use the following notation for the components of $\mathcal{G}(n,n_0;\omega)$, 
\begin{equation}
    \mathcal{G}(n,n_0;\omega) = \begin{pmatrix}
    \G{\up\up}{n} & \G{\up\down}{n} \\
    \G{\down\up}{n} & \G{\down\down}{n}
    \end{pmatrix}
\end{equation}
In the case where $n \neq n_0$, equation \eqref{eq:kgsystem} can be epressed as two (piece-wise) linear system each pertaining to two components of $\mathcal{G}(n,n_0;\omega)$. These systems can be written as, 
\begin{equation}
    \begin{pmatrix} \G{\up j}{n+1} \\ \G{\down j}{n+1} \end{pmatrix} = M_n \begin{pmatrix} \G{\up j}{n} \\ \G{\down j}{n} \end{pmatrix}
\end{equation}
Where $j = \up, \down$ and 
\begin{equation}
    M_n = 
    \begin{cases}
    \begin{pmatrix} 1 & \omega^{-1}z_1(\omega) \\ -z_1(\omega) & 1-\omega^{-1}z_1(\omega)^2 \end{pmatrix} & n \leq 0 \\[1.5em]
    \begin{pmatrix} 1 & \omega^{-1}z_2(\omega) \\ -z_2(\omega) & 1-\omega^{-1}z_2(\omega)^2 \end{pmatrix} & n > 0
    \end{cases}
\end{equation}
 We are now able to solve for the Green's function of $\Kint(n;\omega)$ by simply stitching together the appropriate linear solutions. For any $|\omega| = 1$ the system above is non-degenerate, so we may solve for the eigenvalues and eigenvalues of the system in terms of $\omega$. The eigenvalues of $M_n$ are given by $r_{1,\pm}(\omega)$ when $n\leq 0$ and $r_{2,\pm}(\omega)$ when $n > 0$. Recall these functions are defined in equation \eqref{eq:rdef}. The eigenvectors of $M_n$ are given by $\mathbf{v}_{1,\pm}(\omega)$ when $n \leq 0$ and $\mathbf{v}_{2,\pm}(\omega)$ when $n > 0$. Again these vector valued functions are defined in equation \eqref{eq:vdef}.\\

Before we write the general solution, we must use our boundary condition in the $n$-direction. In this direction, we will require that the entries of the inverse Kasteleyn converge as $n$ becomes large in magnitude. It will suffice to define the following constraint on the Green's function as $n \to \pm \infty$, 
\begin{equation}
    \lim_{n\to\pm\infty} \G{ij}{n} = 0 \label{eq:greenslimit}
\end{equation}
For each $i = \up,\down$ and $j = \up,\down$. In order to write a solution that satisfies the above, we should understand the norm of the eigenvalues as $\omega$ varies along the unit circle. Figure \ref{fig:rootnorms} shows plots of the norms of the roots $\rpm{1}(\omega)$ given various values of $a$ and $b$. Note that when $a=b$, $\rpm{1}(\omega) = \rpm{2}(\omega)$.\\
\begin{figure}[ht]
\centering
\begin{subfigure}
    \centering
    \includegraphics[scale=0.5]{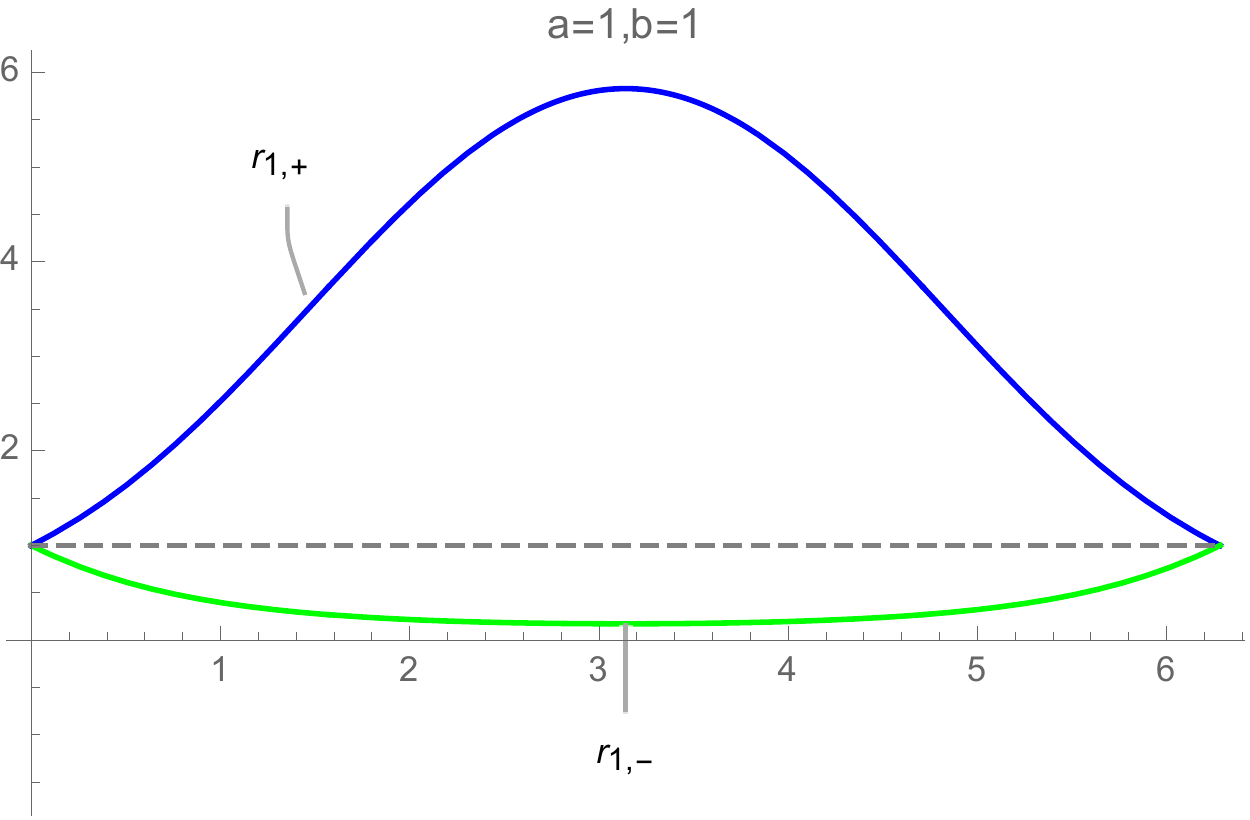}
\end{subfigure}
\begin{subfigure}
    \centering
    \includegraphics[scale=0.5]{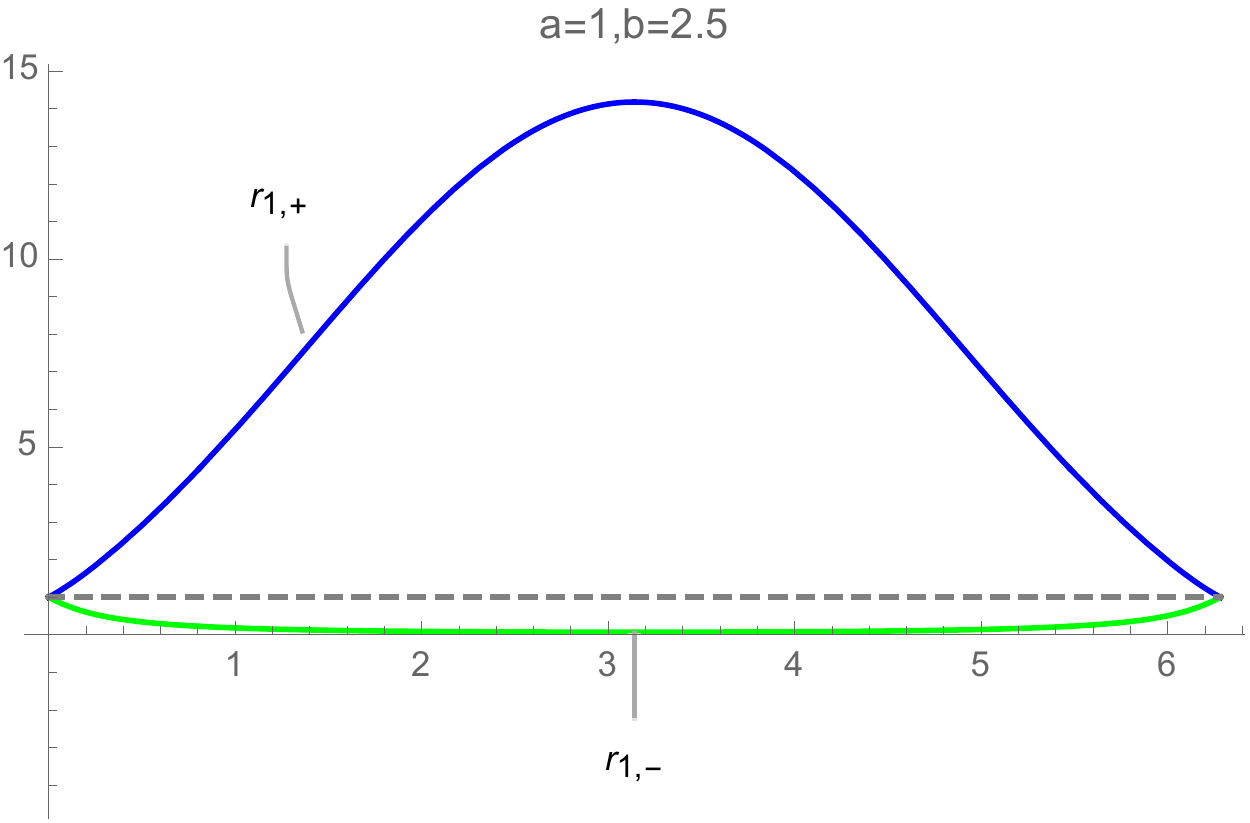}
\end{subfigure}\\[1em]
\begin{subfigure}
    \centering
    \includegraphics[scale=0.5]{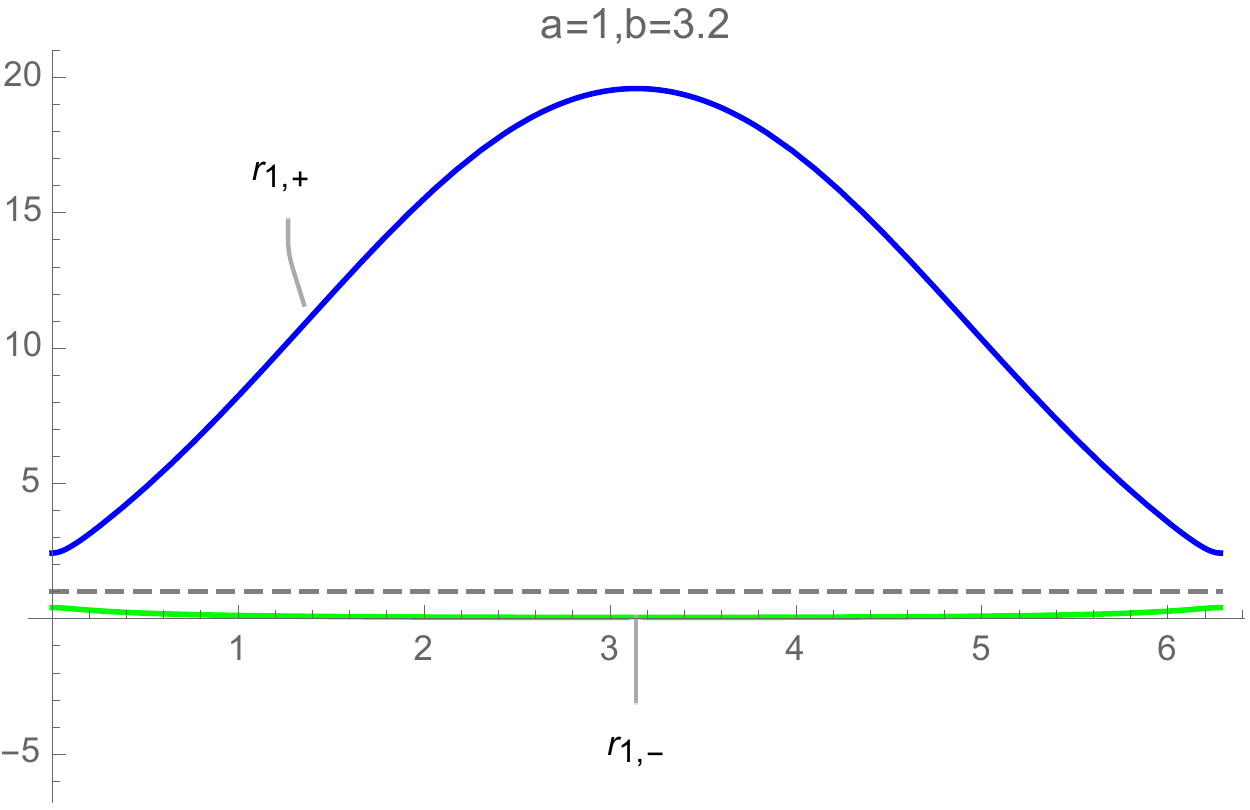}
\end{subfigure}
\begin{subfigure}
    \centering
    \includegraphics[scale=0.5]{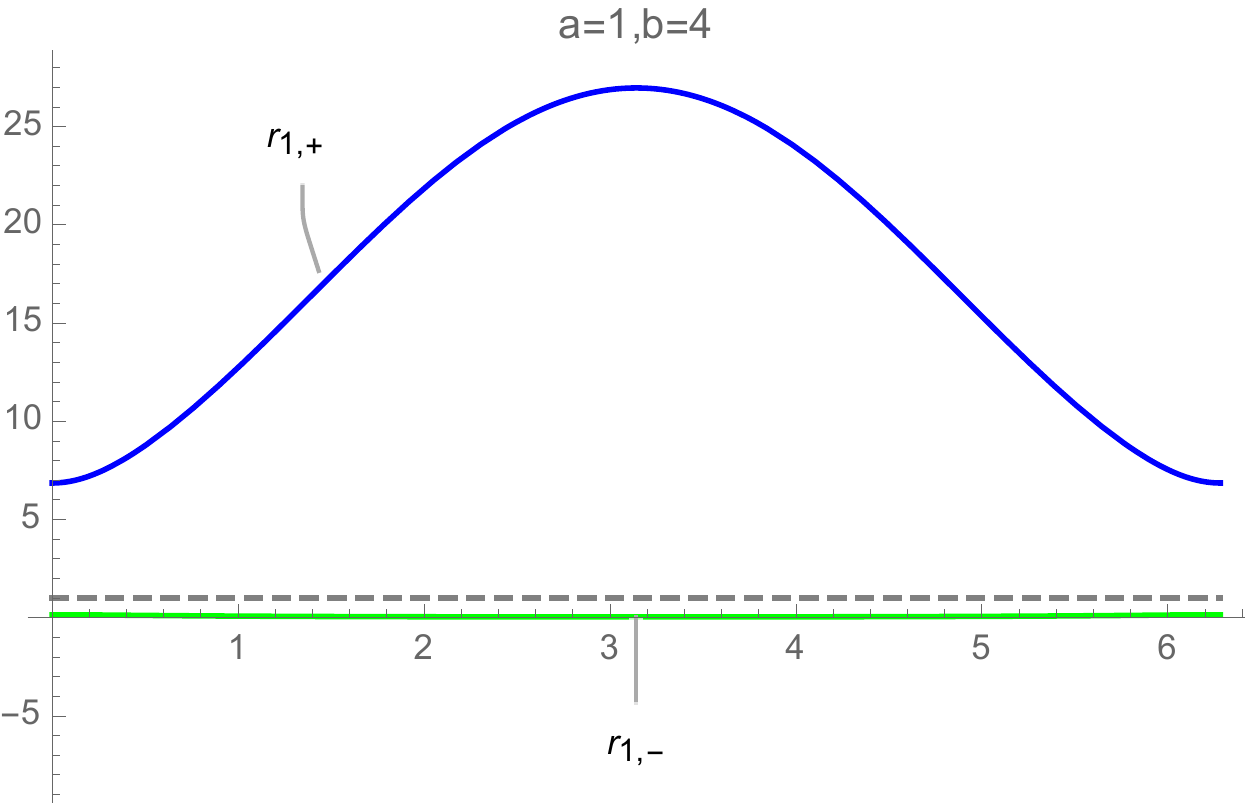}
\end{subfigure}
\caption{\label{fig:rootnorms} The above plots depict the norms of the eigenvalues as $\omega$ varies along the unit circle and for various values of $a$ and $b$. The $x$-axis indicates the argument of $\omega$. The gray dashed line depicts norm equal to $1$. Note that criticality of the weighting can also be seen in these graphs. The system is critical if $|r_{1,+}(\omega)| = |r_{1,-}(\omega)| = 1$ for some $|\omega|=1$. Thus, the the weighting is critical for the top two plots, and non-critical for the bottom two.}
\end{figure}
With all this in mind, we may write the general solution for the Green's function. Since the solution depends on whether $n_0 \leq 0$ or $n_0 > 0$, we will write these solutions are separate cases. We will denote the solutions for when $n_0 \leq 0$ and $n_0 > 0$ by $G_{ij}^<(n,n_0;\omega)$ and $G_{ij}^>(n,n_0;\omega)$, respectively. The general solutions are repeated below, 
\begin{equation*}
    \begin{pmatrix}
    G^>_{\up\up}(n,n_o;\omega) \\ G^>_{\down\up}(n,n_0;\omega)
    \end{pmatrix} = 
    \begin{cases}
    c_1(n_0;\omega) \mathbf{v}_{1,+}(\omega) r_{1,+}(\omega)^n & n \leq 0 \\
    c_2(n_0;\omega) \mathbf{v}_{2,+}(\omega) r_{2,+}(\omega)^n + c_3(n_0;\omega)\mathbf{v}_{2,-}(\omega) r_{2,-}(\omega)^n & 0 < n \leq n_0 \\
    c_4(n_0;\omega) \mathbf{v}_{2,-}(\omega) r_{2,-}(\omega)^n & n > n_0
    \end{cases}
\end{equation*}
\begin{equation*}
    \begin{pmatrix}
    G^>_{\up\down}(n,n_0;\omega) \\ G^>_{\down\down}(n,n_0;\omega)
    \end{pmatrix} = 
    \begin{cases}
    d_1(n_0;\omega) \mathbf{v}_{1,+}(\omega) r_{1,+}(\omega)^n & n \leq 0 \\
    d_2(n_0;\omega) \mathbf{v}_{2,+}(\omega) r_{2,+}(\omega)^n + d_3(n_0;\omega)\mathbf{v}_{2,-}(\omega) r_{2,-}(\omega)^n & 0 < n < n_0 \\
    d_4(n_0;\omega) \mathbf{v}_{2,-}(\omega) r_{2,-}(\omega)(\omega)^n & n \geq n_0
    \end{cases}
\end{equation*}
\begin{equation*}
    \begin{pmatrix}
    G^<_{\up\up}(n,n_0;\omega) \\ G^<_{\down\up}(n,n_0;\omega)
    \end{pmatrix} = 
    \begin{cases}
    c'_1(n_0;\omega) \mathbf{v}_{1,+}(\omega) r_{1,+}(\omega)^n & n \leq n_0 \\
    c'_2(n_0;\omega) \mathbf{v}_{1,+}(\omega) r_{1,+}(\omega)^n + c'_3(n_0;\omega)\mathbf{v}_{1,-}(\omega) r_{1,-}(\omega)^n & n_0 < n \leq 0 \\
    c'_4(n_0;\omega) \mathbf{v}_{2,-}(\omega) r_{2,-}(\omega)^n & n > 0
    \end{cases}
\end{equation*}
\begin{equation*}
    \begin{pmatrix}
    G^<_{\up\down}(n,n_0;\omega) \\ G^<_{\down\down}(n,n_0;\omega)
    \end{pmatrix} = 
    \begin{cases}
    d'_1(n_0;\omega) \mathbf{v}_{1,+}(\omega) r_{1,+}(\omega)^n & n < n_0 \\
    d'_2(n_0;\omega) \mathbf{v}_{1,+}(\omega) r_{1,+}(\omega)^n + d'_3(n_0;\omega)\mathbf{v}_{1,-}(\omega) r_{1,-}(\omega)^n & n_0 \leq n \leq 0 \\
    d'_4(n_0;\omega) \mathbf{v}_{2,-}(\omega) r_{2,-}(\omega)^n & n > 0
    \end{cases}
\end{equation*}
All that is left is to solve for the coefficients. To do this, we consider equation \eqref{eq:kgsystem} when $n = 0$, $n = n_0$, and $n = n_0-1$. All of the coefficients are explicit stated in Appendix \ref{appendixcoef}. \\

Lastly, to get an explicit formula for the entries of $\Kint^{-1}$ we simply use the inverse Fourier transform on the above.
\begin{equation}
    \Kint^{-1}\left(w_i(n_0,0),b_j(n,m)\right) = \frac{1}{2\pi} \int_{|w|=1} G^k_{ij}(n,n_0;\omega) \omega^m \frac{d\omega}{i\omega} \label{eq:invkasteleyn}
\end{equation}
for any $i=\up,\down$, $j=\up,\down$ and $k= >,<$. Note that $k$ depends on the value of $n_0$. We are allowed to fix the vertex $w_i$ at $m_0 = 0$ since the lattice is translationally invariant in the $m$-direction. 
\end{proof}

\section{Asymptotic behavior of the inverse Kasteleyn}
In this section, we will compute the asymptotic behavior of $\Kint^{-1}\left(w_i(n,m),b_j(n_0,0)\right)$ with the goal of understanding the limit shape behavior of the lattice. By determining the decay of the inverse Kasteleyn entries, we will learn whether certain regions of the lattice are critical or non-critical. In the critical case, we expect the decay of the inverse Kasteleyn to be inversely related to the distance, while in the non-critical regions we expect the decay to be exponential.

Throughout this analysis we will assume that the $m_0 = 0$ as the lattice is translationally invariant in the $m$-direction. We will also assume that $b-a > 2$ as we expect this scenario of weights to produce a lattice with both critical and non-critical regions. 
\subsection{Asymptotic behavior for large $m$ and fixed $n$ and $n_0$}
First we will consider the asymptotic behavior when $n$ and $n_0$ are close to the boundary and the vertical distance grows. Regardless of what half of the plane the two vertices are in, the inverse Kasteleyn component decays inversely in $m$. However, it is worth noting that when one of the vertices lies to the left of the interface, it will contribute a term that decays exponentially in the relevant coordinate, even though these values are fixed. We will elaborate on this more after the statement and proof of the corollary,
\begin{corollary}
For $n$ and $n_0$ fixed, the inverse Kasteleyn entries have the following asymptotic expansion,
\begin{equation} \label{eq:corollary1}
    \Kint^{-1}\left(w_i(n_0,0),b_j(n,m)\right) = \frac{1}{\pi m} \Im\left(\lim_{\theta \to 0^+} G^k_{ij}(n,n_0;e^{i\theta})\right) + \mathcal{O}(m^{-2})
\end{equation}
as $m \to \infty$. Where $i = \up,\down$, $j = \up, \down$, and $k = >,<$ depending on the value of $n_0$.
\end{corollary}
\begin{proof}
To start, let's substitute $\omega = e^{i\theta}$ into equation \eqref{eq:invkasteleyn},
\begin{equation} \label{eq:komegasub}
    \Kint^{-1}\left(w_i(n_0,0),b_j(n,m)\right) = \frac{1}{2\pi} \int_0^{2\pi} G^k_{ij}(n,n_0;e^{i\theta}) e^{i\theta m} d\theta
\end{equation}
We wish to employ the method of steepest decent to compute the asymptotics of the above. The coefficient $G^k_{ij}(n,n_0;e^{i\theta})$ is well-behaved on the interval $(0,2\pi)$. We only need to be aware of the jump discontinuity at the endpoints. 

To start we need to deform the contour in the manner depicted by figure \ref{fig:gcontour}. Consider the three integrals that arise from breaking the contour into the two vertical components and the one horizontal component. 
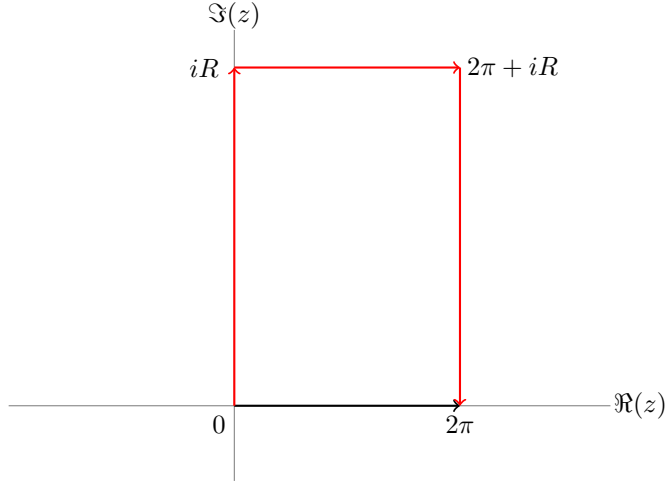
\begin{figure}[h] 
    \centering
    \begin{tikzpicture}
        \draw[gray] (-3,0) -- (5,0);
        \draw[gray] (0,-1) -- (0,5);
        \draw[->,thick] (0,0) -- (3,0);
        \draw[->,thick,red] (0,0) -- (0,4.5);
        \draw[->,thick,red] (0,4.5) -- (3,4.5);
        \draw[->,thick,red] (3,4.5) -- (3,0);
        \node at (5.4,0) {$\Re(z)$};
        \node at (0,5.2) {$\Im(z)$};
        \node at (-0.2,-0.25) {$0$};
        \node at (3,-0.25) {$2\pi$};
        \node at (3.7,4.5) {$2\pi+iR$};
        \node at (-0.4,4.5) {$iR$};
    \end{tikzpicture}
    \caption{\label{fig:gcontour} For the method of steepest descent we deform the original contour (in black) to the contour in red. Note that we let $R \to \infty$.}
\end{figure}
Along the horizontal component, we note that
\begin{equation}
    \Big\|G^k_{ij}(n,n_0;e^{-R}e^{is})e^{-Rm}e^{ism}\Big\| \leq e^{-R}
\end{equation}
so the integrand vanishes as $R \to \infty$. Now we are left to deal with the integrals along the vertical components of the contour. We write the integral along the left most path as, 
\begin{equation}
    \lim_{\theta\to 0^+} \frac{i}{2 \pi} \int_0^\infty G^>_{\up\up}(n,n_0;e^{-s+i\theta})e^{-sm} ds
\end{equation}
and along the right most path as, 
\begin{equation}
    -\lim_{\theta\to 2\pi^-} \frac{i}{2 \pi} \int_0^\infty G^>_{\up\up}(n,n_0;e^{-s+i\theta})e^{-sm} ds
\end{equation}
Note that we need to take limits due to the discontinuity at the endpoints. Since the exponent is real and decaying, we can now use Laplace's method on both these integrals. Immediately we obtain, 
\begin{equation}
    \Kint^{-1}\left(w_i(n_0,0),b_j(n,m)\right) = \frac{1}{2 \pi i m} \left(\lim_{\theta_1 \to 0^+} G^k_{ij}(n,n_0;e^{i\theta_1})- \lim_{\theta_2 \to 2\pi^-} G^k_{ij}(n,n_0;e^{i\theta_2})\right) + \mathcal{O}(m^{-2}) \label{eq:corollary1pre}
\end{equation}
Lastly, we can notice that 
\begin{equation}
    \lim_{\theta_1 \to 0^+} G^k_{ij}(n,n_0;e^{i\theta_1})- \lim_{\theta_2 \to 2\pi^-} G^k_{ij}(n,n_0;e^{i\theta_2}) = 2\Im\left(\lim_{\theta \to 0^+} G^k_{ij}(n,n_0;e^{i\theta})\right)
\end{equation}
and so equation \eqref{eq:corollary1pre} simplifies to what is given in equation \eqref{eq:corollary1}. 
\end{proof}
We would like to take a closer look at the result in Corollary 1 and so it will be useful to consider three specific cases: (1) the case where the vertices both lie to the left of the interface ($n<n_0<0$), (2) the case where the vertices both lie to the right of the interface ($0<n_0<n$), and (3) the case where the vertices lie on opposite sides of the interface ($n_0<0<n$). These cases, and their asymptotics are depicted in figure \ref{fig:closetointerface}. There are certainly more cases we can consider, but looking at these three will be sufficient to understand the behavior around the interface. 
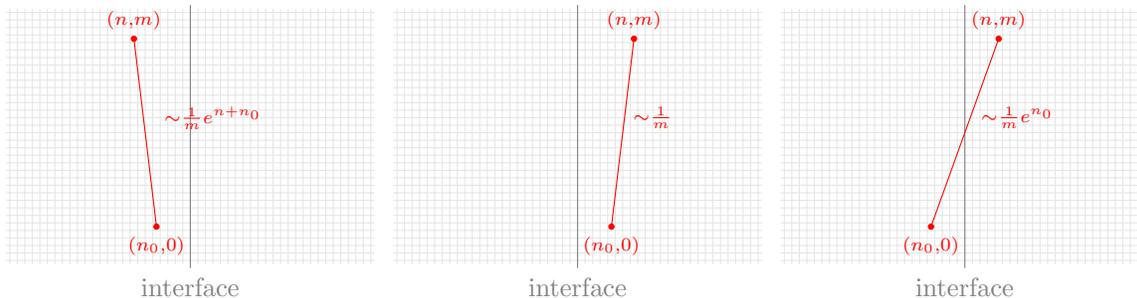
\begin{figure}[ht]
\centering
\begin{subfigure}
    \centering
    \begin{tikzpicture}
        \draw[step=0.1,gray!20,thin] (-2.45,-0.45) grid (2.45,2.95);
        \draw[gray] (0,3) -- (0,-0.5);
        \node[gray] at (0,-0.75) {interface};
        \draw[red] (-0.45,0.05) -- (-0.75,2.55);
        \filldraw[red] (-0.45,0.05) circle (1pt);
        \filldraw[red] (-0.75,2.55) circle (1pt);
        \node[red] at (-0.45,-0.2) {$\scriptstyle (n_0,0)$};
        \node[red] at (-0.75,2.8) {$\scriptstyle (n,m)$};
        \node[red] at (0.3,1.5) {$\scriptstyle \sim \frac{1}{m}e^{n+n_0}$};
    \end{tikzpicture}
\end{subfigure}
\begin{subfigure}
    \centering
    \begin{tikzpicture}
        \draw[step=0.1,gray!20,thin] (-2.45,-0.45) grid (2.45,2.95);
        \draw[gray] (0,3) -- (0,-0.5);
        \node[gray] at (0,-0.75) {interface};
        \draw[red] (0.45,0.05) -- (0.75,2.55);
        \filldraw[red] (0.45,0.05) circle (1pt);
        \filldraw[red] (0.75,2.55) circle (1pt);
        \node[red] at (0.45,-0.2) {$\scriptstyle (n_0,0)$};
        \node[red] at (0.75,2.8) {$\scriptstyle (n,m)$};
        \node[red] at (1,1.5) {$\scriptstyle \sim \frac{1}{m}$};
    \end{tikzpicture}
\end{subfigure}
\begin{subfigure}
    \centering
    \begin{tikzpicture}
        \draw[step=0.1,gray!20,thin] (-2.45,-0.45) grid (2.45,2.95);
        \draw[gray] (0,3) -- (0,-0.5);
        \node[gray] at (0,-0.75) {interface};
        \draw[red] (-0.45,0.05) -- (0.45,2.55);
        \filldraw[red] (-0.45,0.05) circle (1pt);
        \filldraw[red] (0.45,2.55) circle (1pt);
        \node[red] at (-0.45,-0.2) {$\scriptstyle (n_0,0)$};
        \node[red] at (0.45,2.8) {$\scriptstyle (n,m)$};
        \node[red] at (0.7,1.5) {$\scriptstyle \sim \frac{1}{m}e^{n_0}$};
    \end{tikzpicture}
\end{subfigure}
\caption{\label{fig:closetointerface} The three cases we consider when $m$ is large and $n$ and $n_0$ are finite. From left to right the cases are, $n < n_0 < 0$, $0 < n_0 < n$, and $n_0<0<n$.}
\end{figure}

In the first case, where $n<n_0<0$, let's compute the leading order term of the expansion of the $\up\up$ case for $a=1$ and $b=4$,
\begin{equation}
    \Kint^{-1}\left(w_\up(n_0,0),b_\up(n,m)\right) \sim \frac{1}{3m\pi}\left(\frac{-7-3\sqrt{5}}{2}\right)^n\left(\frac{-7+3\sqrt{5}}{2}\right)^{1-n_0}
\end{equation}
And we can see that $n$ and $n_0$ contribute terms that exponentially decay. The analysis follows suit for the other components ($\up\down$, $\down\up$, and $\down\down$) and different values of $a$ and $b$ (so long as $|b-a| >2$).

For the case where $0 < n_0 < n$ we can also compute the first order term of the expansion for the $\up\up$ case where $a = 1$ and $b = 4$. We find, 
\begin{equation}
    \Kint^{-1}\left(w_\up(n_0,0),b_\up(n,m)\right) \sim \frac{1}{3m\pi}
\end{equation}
Since both vertices now lie on the right side of the interface, there are no longer exponentially decaying terms and all decay is inverse in $m$.

And lastly, for the case where $n_0<0<n$ we write out the first order term of the expansion for the $\up\up$ case where $a = 1$ and $b = 4$ as, 
\begin{equation}
    \Kint^{-1}\left(w_\up(n_0,0),b_\up(n,m)\right) \sim \frac{1}{3m\pi}\left(\frac{-7+3\sqrt{5}}{2}\right)^{-n_0}
\end{equation}
We see that the leading order term decays exponentially in $n_0$ but not $n$, since only $n_0$ lies to the left of the interface. It's important to note that although $n_0 < 0$, the term being exponentiated has norm less than $1$ for any value of $a$ and $b$ such that $|b-a|>2$.
\subsection{Asymptotic behavior for large horizontal distances}
Now we would like to comment on the asymptotic behavior of the inverse Kasteleyn matrix when $m$ is fixed and $n$ and $n_0$ are manipulated in a way so that the vertices are far away. There are many ways in which we can accomplish this. In this section, we will consider a few methods to compute the asymptotic behavior for large horizontal distances.

Overall, we find that the asymptotic behavior in the horizontal direction depends on which side of the interface the two vertices lie. For example, if both vertices lie on to the left of the interface, the decay will be exponential in the horizontal distance. However, if both vertices lie to the right of the interface, the decay will be linear in the horizontal distance. 
\subsubsection{Across the interface, $n_0$ fixed.}
First let us consider the case where $n_0$ is fixed, $|n| \to \infty$, and the two vertices are on opposite sides of the interface. This encompasses two cases where one vertex remains close to the interface while the other vertex is far away. These scenarios are shown in figure \ref{fig:acrossintn0fixed}. We expect the asymptotic behavior to be dominated by which side of the interface the vertex indexed by $n$ lies on. 
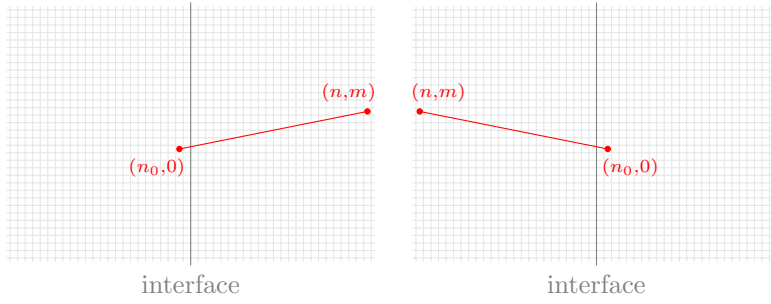
\begin{figure}[ht]
\centering
\begin{subfigure}
    \centering
    \begin{tikzpicture}
        \draw[step=0.1,gray!20,thin] (-2.45,-0.45) grid (2.45,2.95);
        \draw[gray] (0,3) -- (0,-0.5);
        \node[gray] at (0,-0.75) {interface};
        \draw[red] (-0.15,1.05) -- (2.35,1.55);
        \filldraw[red] (-0.15,1.05) circle (1pt);
        \filldraw[red] (2.35,1.55) circle (1pt);
        \node[red] at (-0.45,0.8) {$\scriptstyle (n_0,0)$};
        \node[red] at (2.1,1.8) {$\scriptstyle (n,m)$};
    \end{tikzpicture}
\end{subfigure}
\begin{subfigure}
    \centering
    \begin{tikzpicture}
        \draw[step=0.1,gray!20,thin] (-2.45,-0.45) grid (2.45,2.95);
        \draw[gray] (0,3) -- (0,-0.5);
        \node[gray] at (0,-0.75) {interface};
        \draw[red] (0.15,1.05) -- (-2.35,1.55);
        \filldraw[red] (0.15,1.05) circle (1pt);
        \filldraw[red] (-2.35,1.55) circle (1pt);
        \node[red] at (0.45,0.8) {$\scriptstyle (n_0,0)$};
        \node[red] at (-2.1,1.8) {$\scriptstyle (n,m)$};
    \end{tikzpicture}
\end{subfigure}
\caption{\label{fig:acrossintn0fixed} The two cases we consider where the vertices are located across the interface and $|n|$ is large.}
\end{figure}

Before stating the results, it will be useful to define what I will call little-$g$ functions. These functions are the result of stripping away the $n$ dependence of $G^k_{ij}(n,n;\omega)$ and are only defined for the first and last cases in equations \eqref{eq:g11g21sol1}-\eqref{eq:g12g22sol2}. For example, when $n_0 < 0$ and $n >0$ we define, 
\begin{equation}
    g^<_{ij,n>0}(n_0;\omega) r_{2,-}(\omega)^n = G^<_{ij}(n,n_0;\omega)
\end{equation}
and for $n < n_0 < 0$ we have, 
\begin{equation}
    g^<_{ij,n<n_0}(n_0;\omega) r_{1,+}(\omega)^n = G^<_{ij}(n,n_0;\omega)
\end{equation}
We define the rest of the relevant cases analogously. We will now state two corollaries addressing the scenarios in figure \ref{fig:acrossintn0fixed}. First we have, 
\begin{corollary}
For $n_0<0<n$ and $n_0$ and $m$ fixed the inverse Kasteleyn entries have the following asymptotic expansion,
\begin{equation} \label{eq:asymp2}
    \Kint^{-1}\left(w_i(n_0,0),b_j(n,m)\right) = \frac{1}{2 \pi a n}\left(\lim_{\theta_1 \to 0^+} g^<_{ij,n>0}(n_0;e^{i\theta_1}) + \lim_{\theta_2 \to 2\pi^-} g^<_{ij,n>0}(n_0;e^{i\theta_2}) \right) + \mathcal{O}(n^{-2})
\end{equation}
as $n \to \infty$.
\end{corollary}
\begin{proof}
Not that for any $|\omega|=1$, the root $r_{2,-}(\omega)$ is real. Thus solving for the asymptotics of this integral is a direct application of Laplace's method. Plotting the value's of the root over the unit circle (shown in figure \ref{fig:r2mplot}) shows that there are two maximums of equal value ($r_{2,-}(\omega) = 1$) at $\theta = 0$ and $\theta = 2 \pi$, where we let $\omega = e^{i\theta}$. We must account for the contributions from both when applying Laplace's method.
\begin{figure}[ht]
\centering
\includegraphics[scale=0.5]{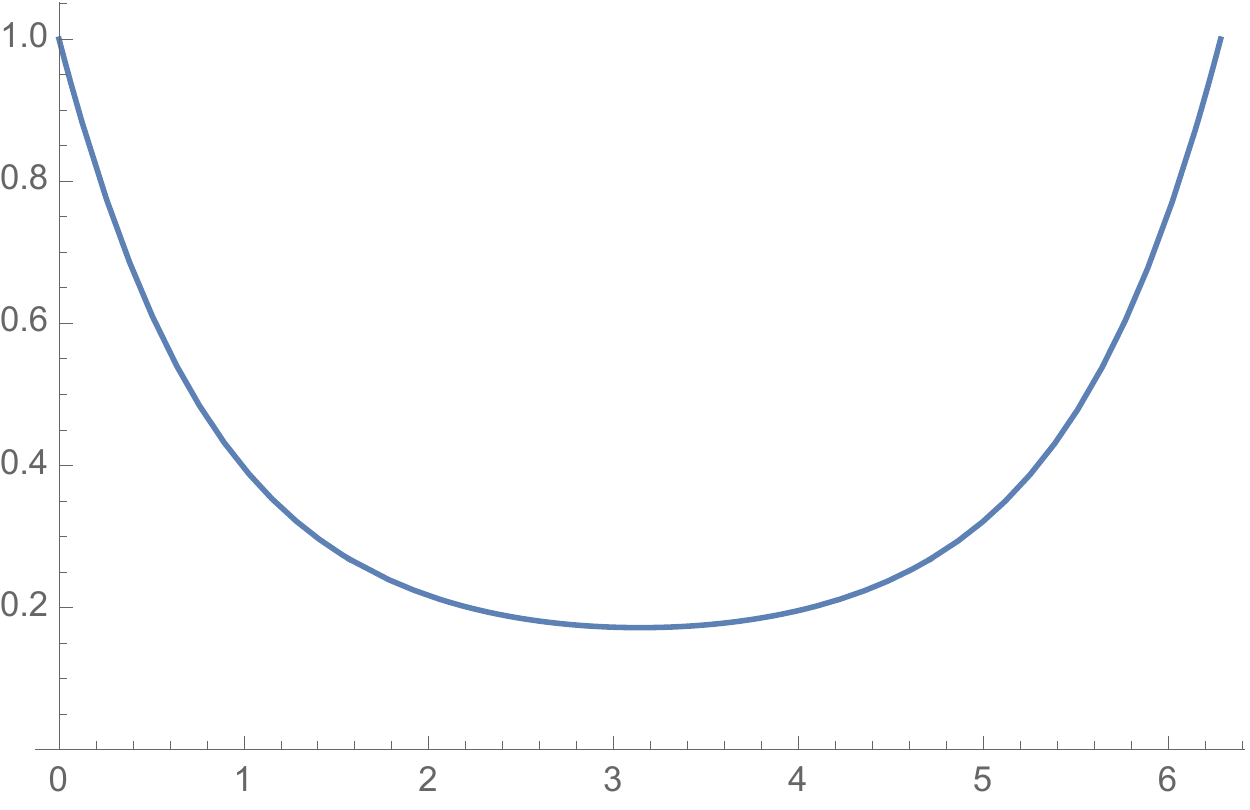}
\caption{\label{fig:r2mplot} Plot of $r_{2,-}(\omega)$ over the unit circle for $a=1$.}
\end{figure}
\end{proof}
The second case consider is when $n < 0 < n_0$ and thus the vertex indexed by $n$ lies to the left of the interface and in the non-critical region. In this instance we have,
\begin{corollary}
For $n < 0 < n_0$ and $n_0$ and $m$ fixed the inverse Kasteleyn operator has the following asymptotic expansion,
\begin{equation} \label{eq:cor3}
    \Kint^{-1}\left(w_i(n_0,0),b_j(n,m)\right) = -\frac{1}{\pi}\Im\left(\lim_{\theta \to 0^+}g^>_{i,j,n<0}\left(n_0,e^{i\theta}\right)\right) \frac{\sqrt{(a-b)^2-4}}{(a+b)n}e^{n r_{1,+}(1)} + \mathcal{O}\left(\frac{e^{nr_{1,+}(1)}}{n^2}\right)
\end{equation}
as $\to -\infty$.
\end{corollary}
In this case, the asymptotic variable $n$ lies to the left of the interface and thus we get exponential decay that depends on $n$ in the leading order term. 
\begin{proof}
We cannot immediately apply Laplace's method because the root $r_{1,+}(e^{i\theta})$ is complex. However, if we deform the contour in the manner shown in figure \ref{fig:gcontour} we will be able to apply Laplace's method appropriately. With this in mind, the proof now follows the methods used in section 5.1.1. First we make sure to write, 
\begin{equation}
    r_{1,+}^n = e^{n \log r_{1,+}}
\end{equation}
where we suppressed the $\omega$ dependance in the notation, and we use the series expansion, 
\begin{equation}
    \log r_{1,+}(e^{-s}) = \log r_{1,+}(1) + \frac{(a+b)}{\sqrt{(a-b)^2-4}}s + \mathcal{O}(s^2)
\end{equation}
From this we get the formula stated above in Corollary 3.
\end{proof}
\subsubsection{Critical side of the lattice, $n_0$ fixed.}
Now let's again consider a set up where $n_0$ is fixed, however now we both vertices will lie to the right of the interface. In particular, $n>n_0>0$. The set up is show in figure \ref{fig:rightofinn0fixed}. As we see in the corollary below, the asymptotic expansion decays inversely with respect to $n$ in the leading order. 
\begin{figure}[ht]
\centering
\begin{tikzpicture}
    \draw[step=0.1,gray!20,thin] (-3.95,-0.45) grid (3.95,2.95);
    \draw[gray] (0,3) -- (0,-0.5);
    \node[gray] at (0,-0.75) {interface};
    \draw[red] (0.15,1.05) -- (3.75,1.55);
    \filldraw[red] (0.15,1.05) circle (1pt);
    \filldraw[red] (3.75,1.55) circle (1pt);
    \node[red] at (0.45,0.8) {$\scriptstyle (n_0,0)$};
    \node[red] at (3.5,1.8) {$\scriptstyle (n,m)$};
\end{tikzpicture}
\caption{\label{fig:rightofinn0fixed} The case where both vertices are to the right of the interface, $n_0$ is finite and $n$ is large.}
\end{figure}
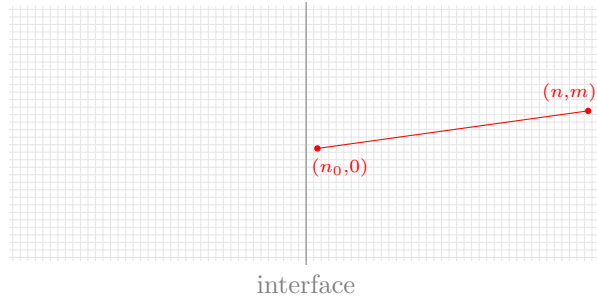
\begin{corollary}
For $n>n_0>0$ and $m$ and $n_0$ are finite, the asymptotic behavior of the inverse Kasteleyn is given by,
\begin{equation}
    \Kint^{-1}\left(w_i(n_0,0),b_j(n,m)\right) = \frac{1}{\pi a n} \left(\lim_{\theta \to 0^+} g^>_{ij,n>n_0}(n_0;e^{i\theta})\right) + \mathcal{O}(n^{-2})
\end{equation}
for $i = \up,\down$ and $j = \up,\down$. 
\end{corollary}
\begin{proof}
Just like in the case of equation \eqref{eq:asymp2}, this is a direct application of Laplace's method.
\end{proof}
\subsubsection{Non-critical side of the lattice, $n_0$ fixed}
We will consider the complement of the case in section 5.2.2, that is we will consider the case where $n<n_0<0$ and $n_0$ is fixed and $-n$ becomes large. This scenario is depicted in figure \ref{fig:leftofinn0fixed}.
\begin{figure}[ht]
\centering
\begin{tikzpicture}
    \draw[step=0.1,gray!20,thin] (-3.95,-0.45) grid (3.95,2.95);
    \draw[gray] (0,3) -- (0,-0.5);
    \node[gray] at (0,-0.75) {interface};
    \draw[red] (-0.15,1.05) -- (-3.75,1.55);
    \filldraw[red] (-0.15,1.05) circle (1pt);
    \filldraw[red] (-3.75,1.55) circle (1pt);
    \node[red] at (-0.45,0.8) {$\scriptstyle (n_0,0)$};
    \node[red] at (-3.5,1.8) {$\scriptstyle (n,m)$};
\end{tikzpicture}
\caption{\label{fig:leftofinn0fixed} The case where both vertices are to the left of the interface, $n_0$ is fixed, and $-n$ is large.}
\end{figure}
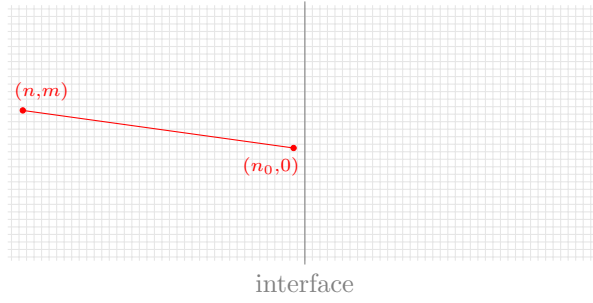
\begin{corollary}
For $n>n_0>0$ and $m$ and $n_0$ are finite, the asymptotic behavior of the inverse Kasteleyn is given by,
\begin{equation}
    \Kint^{-1}\left(w_i(n_0,0),b_j(n,m)\right) = -\frac{1}{\pi}\Im\left(\lim_{\theta \to 0^+}g^<_{i,j,n<n_0}\left(n_0,e^{i\theta}\right)\right) \frac{\sqrt{(a-b)^2-4}}{(a+b)n}e^{n r_{1,+}(1)} + \mathcal{O}\left(\frac{e^{nr_{1,+}(1)}}{n^2}\right)
\end{equation}
for $i = \up,\down$ and $j = \up,\down$. 
\end{corollary}
\begin{proof}
The proof of this corollary follows the same procedure as the proof of equation \eqref{eq:cor3}.
\end{proof}
\subsubsection{Same side of the interface, $n$ and $n_0$ proportional}
Now let's look at the two cases where the vertices are far apart (horizontally) and well to one side of the interface, i.e. neither vertex is fixed by the boundary. The two scenarios of interest are depicted in figure \ref{fig:proportionalcases}.
\begin{figure}[ht]
\centering
\begin{subfigure}
    \centering
    \begin{tikzpicture}
        \draw[step=0.1,gray!20,thin] (-5.45,-0.45) grid (0.45,2.95);
        \draw[gray] (0,3) -- (0,-0.5);
        \node[gray] at (0,-0.75) {interface};
        \draw[red] (-2.15,1.05) -- (-5.35,1.55);
        \filldraw[red] (-2.15,1.05) circle (1pt);
        \filldraw[red] (-5.35,1.55) circle (1pt);
        \node[red] at (-2.45,0.8) {$\scriptstyle (n_0,0)$};
        \node[red] at (-5.1,1.8) {$\scriptstyle (n,m)$};
    \end{tikzpicture}
\end{subfigure}
\begin{subfigure}
    \centering
    \begin{tikzpicture}
        \draw[step=0.1,gray!20,thin] (-0.45,-0.45) grid (5.45,2.95);
        \draw[gray] (0,3) -- (0,-0.5);
        \node[gray] at (0,-0.75) {interface};
        \draw[red] (2.15,1.05) -- (5.35,1.55);
        \filldraw[red] (2.15,1.05) circle (1pt);
        \filldraw[red] (5.35,1.55) circle (1pt);
        \node[red] at (2.45,0.8) {$\scriptstyle (n_0,0)$};
        \node[red] at (5.1,1.8) {$\scriptstyle (n,m)$};
    \end{tikzpicture}
\end{subfigure}
\caption{\label{fig:proportionalcases} The above diagrams depict the cases where $n$ and $n_0$ lie on the same side of the interface and both indices become large at a proportional rate.}
\end{figure}
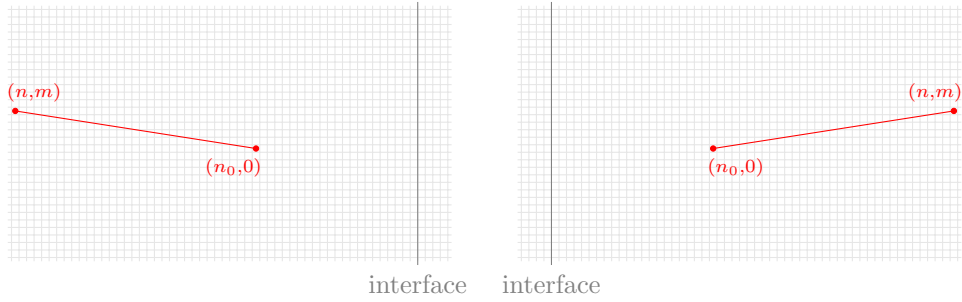
To describe the asymptotic behavior succinctly, it will be useful to define some notation. We define the coefficients $c_{4,1}(\omega)$, $c_{4,2}(\omega)$, $c'_{1,1}(\omega)$ and $c'_{1,2}(\omega)$ by the equations
\begin{equation}
    c_1(\omega) = r_{2,-}(\omega)^{-n_0}c_{4,1}(\omega) + r_{2,+}(\omega)^{-n_0}c_{4,2}(\omega)
\end{equation}
\begin{equation}
    c'_1(\omega) = r_{1,+}(\omega)^{-n_0}c'_{1,1}(\omega)+r_{1,-}(\omega)^{-n_0}c'_{1,2}(\omega)
\end{equation}
One should note that the newly defined coefficients have no dependence on $n_0$. We will first state the corollary describing the asymptotic behavior for when both vertices lie to the right of the interface, in other words the scenario depicted on the right side of figure \ref{fig:proportionalcases}.
\begin{corollary}
For $n_0 = N$, $n = pN$ for some $p > 1$, and $m$ finite, the inverse Kasteleyn operator has the following behavior as $N \to \infty$,
\begin{equation} \label{eq:critical2}
    \Kint^{-1}\left(w_\up(n_0,0),b_\up(n,m)\right) = \frac{1}{a N \pi}\left(\frac{1}{p-1}\lim_{\theta \to 0^+}c_{4,1}(e^{i\theta})z_2(e^{i\theta}) + \frac{1}{p+1}\lim_{\theta \to 0^+}c_{4,2}(e^{i\theta})z_2(e^{i\theta}) \right) + \mathcal{O}(n^{-2})
\end{equation}
\end{corollary}
\begin{proof}
We need to first split the integral up into two, 
\begin{align}
    \Kint^{-1}\left(w_i(n_0,0),b_j(n,m)\right) &= \frac{1}{2\pi} \int_0^{2\pi} c_{4,1}(e^{i\theta})z_2(e^{i\theta})e^{-i\theta}r_{2,-}(e^{i\theta})^{-n_0}r_{2,-}(e^{i\theta})^n e^{i\theta m} d\theta \notag \\
    &\quad\quad\quad\quad+ \frac{1}{2\pi} \int_0^{2\pi} c_{4,2}(e^{i\theta})z_2(e^{i\theta})e^{-i\theta}r_{2,+}(e^{i\theta})^{-n_0}r_{2,-}(e^{i\theta})^n e^{i\theta m} d\theta
\end{align}
and rewrite them in the following manner, 
\begin{align*}
    \int_0^{2\pi} c_{4,1}(e^{i\theta})z_2(e^{i\theta})e^{-i\theta}r_{2,-}(e^{i\theta})^{-n_0}r_{2,-}(e^{i\theta})^n e^{i\theta m} d\theta &=  \int_0^{2\pi} c_{4,1}(e^{i\theta})z_2e^{i(m-1)\theta}\exp\left(N(p-1)\log r_{2,-}\right)  d\theta\\ 
    \int_0^{2\pi} c_{4,2}(e^{i\theta})z_2(e^{i\theta})e^{-i\theta}r_{2,+}(e^{i\theta})^{-n_0}r_{2,-}(e^{i\theta})^n e^{i\theta m} d\theta &= \int_0^{2\pi} c_{4,2}(e^{i\theta})z_2e^{i(m-1)\theta}\exp(-N \log r_{2,+} + pN \log r_{2,-})  d\theta
\end{align*}
In both the the above integrals, the function in the exponent has constant imaginary part and the real part decays appropriately. Thus we may use Laplace's method on the two integrals and the result is given by the asymptotic expansion in equation \eqref{eq:critical2}.
\end{proof}
As expected, corollary 6 shows us the the inverse Kasteleyn entries decay inversely in $N$. Now for the case seen in the left most diagram of figure \ref{fig:proportionalcases} we state the following corollary,
\begin{corollary}
For $n_0 = -N$, $n = -pN$ for some $p > 1$, and $m$ finite, the inverse Kasteleyn operator has the following behavior as $N \to \infty$,
\begin{align} 
    \Kint^{-1}\left(w_\up(n_0,0),b_\up(n,m)\right) &= -\frac{\sqrt{(a-b)^2-4}}{\pi(a+b)(1+p)N}\exp(-pN \log r_{1,+}(1)+N\log r_{1,-}(1))\left(\lim_{\theta\to0^+}c'_{1,2}(e^{i\theta})z_1(e^{i\theta})\right) \notag\\
    &\quad\quad\quad\quad\quad\quad +\mathcal{O}\left(\frac{\exp(-pN \log r_{1,+}(1)+N\log r_{1,-}(1))}{N^2}\right)
\end{align}
\end{corollary}
\begin{proof}
Similarly to the work shown in corollary 6, we can expand the one integral into two integrals,
\begin{align*}
    \frac{1}{2\pi}\int_0^{2 \pi} G^<_{\up\up}(n,n_0;e^{i\theta})e^{im\theta} d\theta =& \frac{1}{2\pi}\int_0^{2\pi} c'_{1,1}(e^{i\theta})z_1(e^{i\theta})e^{i(m-1)\theta}\exp(-N(p-1)\log r_{1,+}(e^{i\theta})) d\theta\\
    &\quad\quad + \frac{1}{2\pi}\int_0^{2\pi} c'_{1,2}(e^{i\theta})z_1(e^{i\theta})e^{i(m-1)\theta}\exp(-pN \log r_{1,+}(e^{i\theta})+N\log r_{1,-}(e^{i\theta})) d\theta
\end{align*}
However these exponents are complex, and so we have to deform the contour in the plane. The contour given in figure \ref{fig:gcontour} will again work. When we deform the first integral on the right side above, the integrals sum to zero. Thus, 
\begin{align*}
    \frac{1}{2\pi}\int_0^{2 \pi} G^<_{\up\up}(n,n_0;e^{i\theta})e^{im\theta} d\theta =& \frac{1}{2\pi}\int_0^{2\pi} c'_{1,2}(e^{i\theta})z_1(e^{i\theta})e^{i(m-1)\theta}\exp(-pN \log r_{1,+}(e^{i\theta})+N\log r_{1,-}(e^{i\theta})) d\theta
\end{align*}
Computing the asymptotic expansion of this follows the work detail in corollary 3. 
\end{proof}
\section{Concluding remarks}
The square lattice with interface allows us to study the asymptotic behavior for a partially non-periodic example of a dimer model. In this instance, we find that the asymptotics of the inverse Kasteleyn behaves predictably as a combination of the two weightings found on either side of the interface. We expect, though it needs to be shown, that the combination of different and or multiple weights in an analogous set up would produce similarly predictable results.

While this work presents a local investigation of the limit shape behavior for the lattice, to gain further understanding of the global behavior it would be worthwhile to investigate the height function (and its asymptotics) of the lattice. One could also investigate whether or not said height function converges to a Gaussian free field in the scaling limit. 

Much like the work in \cite{Kenyon02}, this work can be related to a similar continuous problem regarding the Dirac operator on the plane. In a Dirac type set up, one should consider the mass of the operator changing value or vanishing when $x > 0$. In this continuous setting, I have achieved similarly predictable results to the ones presented here. 

There are many extensions of this work that are worth noting and that, to our knowledge, have not been investigated. Firstly, we would like to consider the interface on more arbitrary bounded domains. In the case of bounded domains, it is unclear whether the limit shape assumes a similar predictable form or achieves something novel. Secondly, one can consider fixing boundary conditions on the interface or at the ends of the lattice and the affect, if any, this has on the limit shape behavior. 
\appendix
\section{Coefficients} \label{appendixcoef}
Below we explicitly state the coefficients present in the Green's function of the Kasteleyn operator with interface that is given in equations \eqref{eq:g11g21sol1}-\eqref{eq:g12g22sol2}. While most of the functions presented here are dependent on the parameter $\omega$, we drop the notation to fit the equations nicely on the page. 
\begin{align}
    c_1 &= \frac{(1 - \rmi{2}) \rp{2}^{-n_0} \omega}{\rp{1}\z{1}(-1+\rmi{2}) + \rmi{2}\z{2}(1-\rp{1})} \\
    c_2 &= \frac{(-1 + \rmi{2}) \rp{2}^{-n_0} \omega}{\z{2}(\rp{2} - \rmi{2})} \\
    c_3 &= \frac{(-1 + \rmi{2}) \rp{2}^{-n_0}\omega \left(\rp{1}\z{1}(1-\rp{2}) + \rp{2}\z{2}(-1 + 
   \rp{1})\right)}{(\rmi{2} - \rp{2}) \z{2}\left(\rp{1}\z{1}(1 - \rmi{2}) + \rmi{2}\z{2}(-1 + \rp{1})\right)} \\
   c_4 &= \frac{\rmi{2}^{-n0} (-1 + \rp{2}) \omega}{\z{2} (\rp{2}-\rmi{2})} + c_3
\end{align}
\begin{align}
    d_1 &= -\frac{\rmi{2} \rp{2}^{1 - n_0} \z{2}}{\rp{1}\z{1}(1 - \rmi{2})+ \rmi{2}\z{2}(-1 + \rp{1})}\\
    d_2 &= \frac{\rmi{2} \rp{2}^{1 - n_0}}{\rmi{2} - \rp{2}}\\
    d_3 &= \frac{\rmi{2} \rp{2}^{1 - n_0} \left(\rp{1}\z{1}(-1 +\rp{2}) + \rp{2}\z{2}(1 - \rp{1})\right)}{(\rp{2} - 
    \rmi{2})\left(\rp{1}\z{1}(-1 + \rmi{2}) + \rmi{2}\z{2}(1 - \rp{1})\right)}\\
    d_4 &= \frac{\rp{2} \rmi{2}^{-n_0 + 1}}{\rmi{2} - \rp{2}} + 
 d_3
\end{align}
\begin{align}
    c'_1 &= \frac{(-1 + \rmi{1}) \rp{1}^{-n_0} \omega}{(\rp{1} - \rmi{1}) \z{1}} + c'_2\\
    c'_2 &= \frac{\rmi{1}^{-n_0} (-1 + \rp{1})\omega\left(\rmi{1}\z{1}(-1 + \rmi{2}) + \rmi{2}\z{2}(1 - \rmi{1})\right)}{(\rmi{1} - \rp{1})\z{1} \left(\rp{1}\z{1}(-1 + \rmi{2}) + \rmi{2}\z{2}(1 - \rp{1})\right)}\\
    c'_3 &= \frac{\rmi{1}^{-n_0} (-1 + \rp{1}) \omega}{(\rp{1} - \rmi{1}) \z{1}}\\
    c'_4 &= \frac{\rmi{1}^{-n_0} (-1 + \rp{1}) \omega}{\rp{1}\z{1}(1-\rmi{2}) + \rmi{2}\z{2}(-1 + \rp{1})}
\end{align}
\begin{align}
    d'_1 &= \frac{\rmi{1} \rp{1}^{1 - n0}}{\rmi{1} - \rp{1}} + d'_2\\
    d'_2 &= \frac{\rmi{1}^{1 - n_0}
   \rp{1} \left(\rmi{1}\z{1}(-1 + \rmi{2}) + \rmi{2}\z{2}(1 - \rmi{1})\right))}{(\rp{1} - 
    \rmi{1}) \left(\rp{1}\z{1}(-1 + \rmi{2}) + \rmi{2}\z{2}(1 - \rp{1})\right)}\\
    d'_3 &= \frac{\rmi{1}^{1 - n_0} \rp{1}}{\rmi{1} - \rp{1}}\\
    d'_4 &= \frac{\rmi{1}^{1 - n_0} \rp{1} \z{1}}{\rp{1}\z{1}(-1+\rmi{2}) + \rmi{2}\z{2}(1 - \rp{1})}
\end{align}
\printbibliography
\end{document}